\documentclass[12pt,a4paper]{article}

\usepackage{cancel}
\usepackage[normalem]{ulem}
\usepackage[T1]{fontenc}
\usepackage[utf8]{inputenc}
\usepackage{authblk}
\usepackage{graphicx,booktabs,balance}
\usepackage{tikz}
\usepackage{algorithmicx}
\usepackage{algpseudocode}
\usepackage[english]{babel}
\usepackage{color}
\usepackage{hyperref}
\usepackage{amsmath,amssymb,amsthm}
\usepackage{caption}
\usepackage{listings}

\newtheorem{theorem}{Theorem}[section]
\newtheorem{corollary}{Corollary}

\newtheorem{lemma}[theorem]{Lemma}
\newtheorem{proposition}{Proposition}

\theoremstyle{definition}
\newtheorem{definition}[theorem]{Definition}
\newtheorem{remark}{Remark}

\newtheorem{example}{Example}

\newcommand{\bbF}{{\mathbb{F}}}
\newcommand{\lexprec}{{\prec_{\mathrm{Lex}}}}
\newcommand{\lexpreceq}{{\preceq_{\mathrm{Lex}}}}
\newcommand{\antiprec}{{\prec_{\mathrm{A }}}}
\newcommand{\antisucc}{{\succ_{\mathrm{A }}}}
\newcommand{\antisucceq}{{\succeq_{\mathrm{A }}}}
\newcommand{\antipreceq}{{\preceq_{\mathrm{A }}}}
\newcommand{\Psucceq}{{\succeq_{\mathrm{P }}}}
\newcommand{\Ppreceq}{{\preceq_{\mathrm{P }}}}

\title{Relative generalized Hamming weights of $q$-ary Reed-Muller codes}
\author{Olav Geil\thanks{olav@math.aau.dk}}
\author{Stefano Martin\thanks{stefano@math.aau.dk}}
\affil{Department of Mathematical Sciences, Aalborg University, Denmark}

\begin{document}
\maketitle

\begin{abstract}
Coset constructions of $q$-ary Reed-Muller codes can be used to store
secrets on a distributed storage system in such a way that only
parties with access to a large part of the system can obtain
information while still allowing for local error-correction. In this paper
we determine the relative generalized Hamming weights of these codes
which can be translated into a detailed description of the information leakage~\cite{Bains,luo2005some,kurihara2012secret,geil2014}.\\

\end{abstract}

\section{Introduction}\label{secosaco1}

We consider the situation where a central party wants to store sensitive
information (a secret) on a distributed storage system in such a way
that other parties with access to a large part of the system will be
able to recover it, but other parties will not. The following requirements are natural:
\begin{itemize}
\item[R1:] Access to arbitrary $r$ (or more) of the stored data
  symbols makes it possible to recover the secret in full, however, with 
  $\tau$ (or less) one cannot recover any information -- or less
  restrictive one can only recover a limited amount of information.
\item[R2:] The storage device must be able to locally repair
  itself. More precisely, if the storage media experiences random
  errors then with a very high probability any stored
  symbol can be corrected from only a small number of randomly
  accessed locations (symbols) of the media.
\end{itemize}
To meet simultaneously the requirements R1 and R2 we propose to use
a coset construction $C_1/C_2$ of $q$-ary Reed-Muller codes. As is
well-known any linear ramp secret sharing scheme can be realized as a
coset construction of two linear codes and vice versa~\cite{chen2007secure}. By
choosing the code $C_1$ to be a $q$-ary Reed-Muller code not only do
we address R1 but we also meet the requirement R2. This is due to the
fact that $q$-ary
Reed-Muller codes are locally
correctable. When
considering the coset construction $C_1/C_2$ rather than $C_1$ this property is
always maintained (See Section~\ref{secosaco2}) which corresponds to R2. The local correctability
properties of $q$-ary Reed-Muller codes have been studied in 
detail, see e.g.\ \cite{58,now}. 

\begin{definition} \label{defi00}
Let $q$ be a power of a prime, $u$ an integer, $s$
a positive integer, and write $n=q^s$. We enumerate the elements of $({\mathbb{F}}_q)^s$ as
$\{ P_1, \ldots , P_n\}$ and consider the
evaluation map $\varphi: \bbF_q[X_1,\dots,X_s] \rightarrow
(\bbF_q)^n$, $\varphi(f) = (f(P_1),\ldots,f(P_n))$. 
The $q$-ary Reed-Muller code of order $u$ in $s$ variables is defined by
\begin{eqnarray}
&RM_q(u,s)  =  \{\varphi(f): f \in \bbF_q[X_1,\ldots,X_s], \deg(f)
\leq u\}.
\end{eqnarray}
\end{definition}

\begin{definition}
A code $C \subseteq ({\mathbb{F}}_q)^n$ is said to be $(\rho, \delta,
\varepsilon)$-correctable if there exists a randomized
error-correcting algorithm ${\mathcal{A}}$ which takes as input
$\vec{y} \in ({\mathbb{F}}_q)^n$ and $i \in \{1,\ldots , n\}$ such that
\begin{enumerate}
\item for all $i \in \{1, \ldots , n\}$ and all vectors $\vec{c} \in C$, $\vec{y}
  \in ({\mathbb{F}}_q)^n$ which differ in at most $\delta$ positions
$$Pr[{\mathcal{A}}(\vec{y},i)=c_i] \geq 1-\varepsilon$$
where the probability is modelling the random coin tosses of the
algorithm ${\mathcal{A}}$. Here, $c_i$ means the $i$-th entry of
$\vec{c}$.
\item ${\mathcal{A}}$ makes at most $\rho$ queries to $\vec{y}$. 
\end{enumerate}
\end{definition}
The following theorem corresponds to~\cite[Pro.\ 2.4, Pro.\ 2.5,
  Pro.\ 2.6]{now}.

\begin{theorem}\label{thelocal}
If $u < q-1$ then 
$RM_q(u,s)$ is $(u+1,\delta, (u+1)\delta)$-locally correctable for all
  $\delta$. Let $\sigma <1$ be a positive real and assume $u
  \leq \sigma (q-1)-1$. Then $RM_q(u,s)$ is
  $(q-1,\delta,2\delta/(1-\sigma))$-locally correctable for all
  $\delta$ and if furthermore $\delta < 1/2-\sigma$ then it is 
$(q-1,\delta,4(\delta -
  \delta^2)/[(q-1)(1-2(\sigma-\delta))^2])$-locally correctable.
\end{theorem}

Turning to the question of information leakage in connection with ramp
secret sharing schemes based on $q$-ary Reed Muller codes, not much
can be found in the literature (see, however, \cite{duursma2012multiplicative} for other interesting
results on secret sharing schemes related to binary Reed-Muller
codes). In the present paper we fill this gap. More precisely we
establish the true values of all corresponding relative generalized
Hamming weights for $q$-ary Reed Muller codes in two variables. For
the case of more variables we device a simple and low complexity
algorithm to determine the parameters. By known methods these results
then easily translate into a detailed and accurate description of
the leakage to unauthorized parties as well as the number of
symbols needed for the authorized parties to recover the secret. We
note that a similar analysis has not been made before for any of the
known families of locally correctable codes.\\

Our work on relative generalized Hamming weights of $q$-ary
Reed-Muller codes is a non-trivial generalization of results by
Heijnen and Pellikaan~\cite{heijnen1998generalized}, who based on the
Feng-Rao bound for dual codes, showed how to calculate generalized Hamming
weights of $q$-ary Reed-Muller codes. Until recently the relative generalized Hamming weights have been
determined. 
for one
family only, namely the family of
MDS-codes. In the recent paper~\cite{geil2014} a method was given to estimate these parameters for one-point algebraic geometric
codes through the use of the Feng-Rao bounds for dual or primary codes. More results in this
direction were presented in~\cite{tower}, ~\cite{geilrefined}, ~\cite{zhang2015relative} and \cite{lee}, the latter
dealing with more-point algebraic geometric codes. The
present paper is a natural continuation
of~\cite{heijnen1998generalized} and \cite{geil2014}, however, to keep the description
as simple as possible, in the presentation of the present paper we use
the footprint bound from Gr\"{o}bner basis theory rather than the
Feng-Rao bounds.\\

The paper is organized as follows. We start in Section~\ref{secosaco2} by
giving some background information on ramp secret sharing schemes and in
particular we explain the role of a coset construction of $q$-ary
Reed-Muller codes. The subsequent four sections treat our main task which
is determination of the relative generalized Hamming weights of
$q$-ary Reed-Muller codes. In Section~\ref{secosaco3} we present the
theory, based on which we shall derive the
weights. Section~\ref{secosaco4} shows
a general method
to derive any of the weights, and this method is formalized into a simple and
low complexity algorithm in Section~\ref{secosaco5}. Finally in 
Section~\ref{secosaco6} we present closed formula expressions for
$q$-ary Reed-Muller codes in two variables. In Section~\ref{secosaco7}
we revert to the communication problem of secret sharing on a
distributed storage system. We make some general remarks on the
connection between information leakage and local correctability and we
give a number of examples. Section~\ref{secosaco8} is the
conclusion. 

\section{Linear ramp secret sharing schemes}\label{secosaco2}

A ramp secret sharing scheme is a cryptographic method to encode a secret
$\vec{s}$ into multiple shares $c_1, \ldots , c_n$ so that only from
specified subsets of the shares one can recover
$\vec{s}$. The encoding is in general probabilistic, meaning that to each secret
$\vec{s}$ there corresponds a collection of possible share vectors
$\vec{c}=(c_1, \ldots , c_n)$. Special attention has been given to linear ramp secret sharing
schemes \cite{chen2007secure}. Here, the space of secrets is $({\mathbb{F}}_q)^\ell$, where
$\ell \geq 1$ is some fixed integer, and $c_1, \ldots , c_n \in
{\mathbb{F}}_q$. Moreover, if $\vec{c}_1$ is an encoding of
$\vec{s}_1$ and $\vec{c}_2$ is an encoding of $\vec{s}_2$, then also
$\vec{c}_1+\vec{c}_2$ is an encoding of
$\vec{s}_1+\vec{s}_2$.\\

A linear ramp secret sharing scheme with $n$ participants, secrets in
$({\mathbb{F}}_q)^\ell$, and shares
belonging to ${\mathbb{F}}_q$ can be described as
follows~\cite{chen2007secure}. Consider linear codes $C_2 \subsetneq C_1 \subset
({\mathbb{F}}_q)^n$ with $\ell = \dim (C_1) - \dim (C_2)$ and let $L \subseteq ({\mathbb{F}}_q)^n$ be (a
linear code) such that $C_1=L \oplus C_2$, where $\oplus$ is the
direct sum. Consider a
vector space isomorphism $\psi: (\bbF_q)^\ell \rightarrow L$. A secret
$\vec{s} \in ({\mathbb{F}}_q)^\ell$ is mapped to
$\vec{x}=\psi(\vec{s})+\vec{c}_2\in C_1$, where $\vec{c}_2 \in C_2$ is chosen by
random. For the analysis, we assume that also the secrets are chosen uniformly from $(\bbF_q)^\ell$. 
In this way, the vectors of shares are chosen uniformly from $C_1$.
The $n$ shares distributed among the $n$ participants are the
$n$ coordinates of $\vec{x}$.  The threshold parameters $t$ and $r$ of the scheme
are the unique numbers such that: 
\begin{enumerate}
\item No group of $t$ participants can
recover any information about $\vec{s}$, but some groups of size $t+1$
can.
\item All groups of size $r$ can recover the secret in full, but some
groups of size $r-1$ cannot. 
\end{enumerate}
Only for $\ell=1$ we can hope for $r=t+1$
in which case we have a complete picture of the security. Such schemes
are called $t$-threshold secret sharing schemes. For general linear
ramp secret sharing schemes we have the parameters $t_1, \ldots ,
t_\ell , r_1, \ldots , r_\ell $ where for $m=1, \ldots , \ell$,  $t_m$
and $r_m$ are the unique numbers such that the following hold:
\begin{enumerate}
\item  No
group of $t_m$ participants can recover $m$ $q$-bits of information
about $\vec{s}$, but some groups of size $t_m+1$ can.
\item  All groups of
size $r_m$ can recover $m$ $q$-bits of information about $\vec{s}$,
but some groups of size $r_m-1$ cannot. 
\end{enumerate}
Clearly, $t=t_1$ and $r=r_\ell$. Observe that the $\tau$ in
requirement R1 could either be $t$ or it could be $t_i$ for some low
value of $i$. 
From~\cite[Th.\ 6.7]{Bains},
\cite[Th.\ 4]{kurihara2012secret} and \cite[Th.\ 6]{geil2014} we have the
following characterization of these parameters:
\begin{eqnarray}
t_m&=&M_m((C_2)^\perp,(C_1)^\perp)-1 \label{eqsnabeleins} \\
r_m&=&n-M_{\ell-m+1}(C_1,C_2)+1, \label{eqsnabeltwei} 
\end{eqnarray}
where $M_m(C_1,C_2)$ is the $m$-th relative generalized Hamming weight
for $C_1$ with respect to $C_2$ and $(C)^\perp$ denotes the dual code
of $C$. To make the section complete we need a formal definition of
these parameters. We start by recalling the well-know concept of generalized Hamming
weights \cite{klove1978weight, helleseth1977weight,wei1991generalized}. Recall that for $D \subseteq
(\bbF_q)^n$ the support of $D$ is defined as $$\mathrm{supp}(D) =
\{i :  c_i \neq 0\mbox{ for some }\vec{c}=(c_1, \ldots , c_n) \in D\}.$$

\begin{definition} \label{defi2}
Let $C$ be a linear code and $k$ its dimension. For $r=1,\ldots,k$, the $r$-th generalized Hamming weight (GHW) of $C$ is defined by
$$d_r(C) = \min\{|\mathrm{supp}(D)| : D\mbox{ is a linear subcode of }C\mbox{ and }\dim(D) = r\}.$$
The sequence $(d_1(C),\ldots,d_{k}(C))$ is called the hierarchy of the GHWs of $C$.
\end{definition}

Note that in particular $d_1(C)$ is the minimum distance of $C$. A further generalization of GHWs was introduced by Luo et al. in \cite{luo2005some}.

\begin{definition} \label{defi3}
Let $C_2 \subsetneq C_1$ be linear codes, $\ell = \dim(C_1)-\dim(C_2)$ the codimension of $C_1$ and $C_2$, and $n$ the length of the codes. For $m =1,\ldots,\ell$, the $m$-th relative generalized Hamming weight (RGHW) of $C_1$ with respect to $C_2$ is defined by
$$M_m(C_1,C_2) = \min_{J \subseteq \{1,\ldots,n\}} \{|J| : \dim((C_1)_J) - \dim((C_2)_J) = m\}$$
where $(C_i)_J = \{\vec{c} \in C_i : c_t = 0\mbox{ for }t \notin J\}$ for $i=1,2$.
The sequence $(M_1(C_1,C_2),\ldots,$ $M_{\ell}(C_1,C_2))$ is called the hierarchy of the RGHWs of $C_1$ with respect to $C_2$.
\end{definition}

If $C_2$ is the zero code $\{ \vec{0}\}$ then the $m$-th RGHW of $C_1$ with respect to $C_2$ is equivalent to the $m$-th GHW of $C_1$. This fact should be more clear from the following result \cite[Lem.\ 1]{liu2008relative}.

\begin{theorem} \label{teo1}
Let $C_2 \subsetneq C_1$ be linear codes and $\ell = \dim(C_1)-\dim(C_2)$ be the codimension of $C_1$ and $C_2$. For $m =1,\ldots,\ell$ we have that
\begin{eqnarray}
M_m(C_1,C_2) = \min\{|\mathrm{supp}(D)| : D\mbox{ is a linear
  subcode of }C_1, {\mbox{ \ \ \ }} \nonumber \\
D \cap C_2 = \{\vec{0}\}\mbox{ and }\dim(D) = m\}.\nonumber
\end{eqnarray}

\end{theorem}

This alternative characterization of RGHWs is useful when the codes
are of an algebraic nature.

\begin{remark}\label{remduality}
It is well-known that 
$RM_q(u,s)^\perp=RM_q((q-1)m-u-1,s)$, \cite[Rem.\ 4.7]{heijnen1998generalized}. Hence,
when both $C_1$ and $C_2$ are $q$-ary Reed-Muller codes then the
information leakage described in~(\ref{eqsnabeleins}) and (\ref{eqsnabeltwei}) is all about relative
generalized Hamming weights of $q$-ary Reed-Muller codes.
\end{remark}

\begin{remark}
As described in this section, in a ramp secret sharing scheme $C_1/C_2$, the code
$C_1$ is divided into disjoint subsets each corresponding to a given
message. The security comes from the randomness with which one picks
the element of the subset. This randomness does not reduce (nor
increase) the locally
error-correcting ability of $C_1$ as the encoded message is still
a word in $C_1$. Hence, if $C_1$ is a $q$-ary Reed-Muller code then
Theorem~\ref{thelocal} describes the ability to perform local
error-correction in $C_1/C_2$.
\end{remark}

We finally remark that the situation of secret sharing is more or less  similar to that of communication
over a wire-tap channel of type II \cite{wyner1975wire}, however we shall not
pursue this connection any further in the present paper.\\

In the following four sections we shall concentrate on estimating the
RGHWs of $q$-ary Reed-Muller codes which as noted gives an overview
on the information leakage from the corresponding schemes $C_1/C_2$.

\section{Useful tools to establish the RGHWs}\label{sec3}\label{secosaco3}
We start our investigations by presenting in this section some theory that shall
help us to derive the weights. The section also includes some initial
results in this direction. First we elaborate slightly on the
definition of $q$-ary Reed-Muller codes.

\begin{definition} \label{defi0}
Let $q$ be a power of a prime, $u$ an integer, $s$
a positive integer, and write $n=q^s$. We enumerate the elements of $({\mathbb{F}}_q)^s$ as
$\{ P_1, \ldots , P_n\}$ and consider the
evaluation map $\varphi: \bbF_q[X_1,\dots,X_s] \rightarrow
(\bbF_q)^n$, $\varphi(f) = (f(P_1),\ldots,f(P_n))$. 
The $q$-ary Reed-Muller code of order $u$ in $s$ variables is defined by
\begin{eqnarray}
&RM_q(u,s)  =  \{\varphi(f): f \in \bbF_q[X_1,\ldots,X_s], \deg(f)
\leq u\} \nonumber \\
& = \mathrm{span}_{\bbF_q}\{\varphi(X_1^{a_1}\cdots X_s^{a_s}) : 
0 \leq a_1,\ldots,a_s < q, a_1+\cdots+a_s \leq u\}. \label{eqrm}
\end{eqnarray}
In this paper we shall use the convention $\deg(0) = -1$ and
$\mathrm{span}_{\bbF_q}\{ \}=\{ \vec{0}\}$. Hence $RM_q(-1,s) = \{\vec{0}\}$.
\end{definition}
Throughout the rest of the paper we shall always write $n=q^s$. Observe that the equality in~(\ref{eqrm}) is a consequence of the fact
that 
\begin{eqnarray}
\varphi(f)=\varphi( f {\mbox{ rem }} \{X_1^q-X_1, \ldots ,
X_s^q-X_s\})\label{eqbrugesmaaske}
\end{eqnarray}
for any $f \in {\mathbb{F}}_q[X_1, \ldots , X_s]$. Here, the argument
on the right side of~(\ref{eqbrugesmaaske}) means the remainder of $f$
after division with $\{X_1^q-X_1, \ldots , X_s^q-X_s\}$
(see~\cite[Sec.\ 2.3]{cox2012ideals} for the multivariate division
algorithm). Furthermore note that $\varphi$ is surjective which is seen
by applying Lagrange interpolation. Dimension considerations now show
that the restriction of $\varphi$ to the span of 
$$R_q^s=\{X_1^{a_1} \cdots X_s^{a_s} : 0 \leq a_i <q, i=1, \ldots ,
s\}$$
is a bijection and  $\{ \varphi (M) : M \in R_q^s\}$ therefore is a basis for
$({\mathbb{F}}_q)^n$ as a vector space. We write 
$$Q^s_q=\{(a_1, \ldots ,
a_s) \in {\mathbb{N}}_0^s : 0 \leq a_i < q, i = 1, \ldots , s\}$$
and $\vec{X}^{\vec{a}}=X_1^{a_1}\cdots X_s^{a_s}$ for $\vec{a}=(a_1,
  \ldots , a_s)\in {\mathbb{N}}_0^s$. Hence, $R_q^s=\{\vec{X}^{\vec{a}} : \vec{a} \in
  Q^s_q\}$. 

\begin{remark}\label{rembasisforD}
From the above discussion we conclude that if $D \subseteq
  {\mbox{RM}}_q(u,s)$ is a subspace of dimension $m$ then without loss
  of generality we may assume that
  $D={\mbox{span}}_{\mathbb{F}_q}\{\varphi(F_1), \ldots ,
  \varphi(F_m)\}$ where the leading monomials (with respect to the
  given fixed monomial ordering $\prec$) satisfy ${\mbox{lm}}(F_i) \in
  R_q^s$, ${\mbox{lm}}(F_i) \neq {\mbox{lm}}(F_j)$ for $i \neq j$, and $\deg (F_i) \leq u$ for $i = 1, \ldots m$. For given $D$
  and fixed $\prec$ these leading monomials are unique.
\end{remark}

We could calculate the RGHWs of $q$-ary Reed-Muller codes using 
the technique from~\cite{geil2014} where the Feng-Rao bound for primary
codes is employed. However, the simple algebraic structure of the
$q$-ary Reed-Muller codes suggests that instead we should apply the
footprint bound which we now introduce. 

\begin{definition}\label{deffoot}
Let $k$ be a field and consider an ideal $J \subseteq
k[X_1,\ldots,X_s]$ and a fixed monomial ordering $\prec$. Let
$\mathcal{M}(X_1,\ldots,X_s)$ denote the set of monomials in the variables $X_1,\ldots,X_s$. The footprint of $J$ with respect to $\prec$ is the set
\begin{eqnarray}
\Delta_\prec(J) = \{M \in \mathcal{M}(X_1,\ldots,X_s) : M\mbox{ is
  not leading monomial of any polynomial in }J\}.\nonumber
\end{eqnarray}
\end{definition}
\begin{example}\label{ex47}
We see immediately that $\Delta_\prec(\langle X_1^q-X_1, \ldots ,
X_s^q-X_s\rangle ) \subseteq R_q^s$.
\end{example}
From~\cite[Th.\ 6]{cox2012ideals} we have the following well-known result.
\begin{theorem}\label{thecox}
Let the notation be as in Definition~\ref{deffoot}. The set
$\{M+J : M \in \Delta_\prec(J)\}$ is a basis for $k[X_1, \ldots ,
  X_s]/J$ as a vector space over $k$.
\end{theorem}
\begin{example}
This is a continuation of Example~\ref{ex47}. From
Theorem~\ref{thecox} and the fact th at $\varphi : R_q^s \rightarrow
({\mathbb{F}}_q)^{n}$ is a bijection we conclude
$\Delta_\prec(\langle X_1^q-X_1, \ldots , X_s^q-X_s\rangle )=R_q^s$.
\end{example}

Consider polynomials $F_1, \ldots , F_m \in {\mathbb{F}}_q[X_1, \ldots
  , X_s]$. Let $\{Q_1, \ldots , Q_N\}$ be their common zeros over
${\mathbb{F}}_q$ and define the vector space homomorphism $\psi:
{\mathbb{F}}_q[X_1, \ldots , X_s] \rightarrow ({\mathbb{F}}_q)^N$,
$\psi(f)=(f(Q_1), \ldots , f(Q_s))$. This map is surjective (Lagrange
interpolation again) and by Theorem~\ref{thecox} the domain of $\psi$
  is a vector space of
  dimension $\left|\Delta_{\prec}(\left\langle F_1,\ldots,F_m,
X_1^q-X_1,\ldots,X^q_s - X_s \right\rangle)\right|$ (independently of
  the chosen monomial
  ordering $\prec$). As a corollary to Theorem~\ref{thecox} we therefore
obtain the following incidence of the footprint bound. For
the general version of the footprint bound see~\cite{onorin}
and~\cite[Pro.\ 8, Sec.\ 5.3]{cox2012ideals}.

\begin{lemma} \label{lemfoot}
Let $F_1,\ldots,F_m \in \bbF_q[X_1,\ldots,X_s]$. The number of common
zeros of $F_1,\ldots,F_m$ over $\bbF_q$ is at most equal to 
$\left|\Delta_{\prec}(\left\langle F_1,\ldots,F_m,
X_1^q-X_1,\ldots,X^q_s - X_s \right\rangle)\right|$ 
(here, $\prec$ is any monomial
ordering).
\end{lemma}
We note that actually equality holds in Lemma~\ref{lemfoot}
(see~\cite[Pro.\ 8, Sec.\ 5.3]{cox2012ideals}), but we shall
not need this fact. To make Lemma~\ref{lemfoot} operational we
recall the following notation from \cite{bezrukov2005macaulay}.

\begin{definition}\label{defpart}
The partial ordering $\Ppreceq$ 
on the monomials in $R_q^s$ and on the elements in $Q^s_q$ is defined by
$${\vec{X}}^{\vec{a}} \Ppreceq {\vec{X}}^{\vec{b}}\mbox{ (or }\vec{a} \Ppreceq \vec{b}\mbox{)} \iff a_i \leq b_i\mbox{ for all }i\in \{1,\ldots,s\}.$$
The upward shadow of $\vec{a} \in Q^s_q$ is $\nabla \vec{a} = \{\vec{b} \in Q^s_q :  \vec{b} \Psucceq \vec{a}\}.$\\
The lower shadow  of $\vec{a} \in Q^s_q$ is $\Delta \vec{a} = \{\vec{b} \in Q^s_q :  \vec{b} \Ppreceq \vec{a}\}.$\\
Let $A \subseteq Q^s_q$, we define $\nabla A = \bigcup_{\vec{a} \in A} \nabla \vec{a}$ and $\Delta A = \bigcup_{\vec{a} \in A} \Delta \vec{a}$.
\end{definition}

\begin{example} \label{ex4}
For  $\vec{a} = (2,3) \in Q^2_4$ we have that
$$\nabla \vec{a} = \{(2,3),(3,3)\}$$
$$\Delta \vec{a} = \{(2,3),(1,3),(0,3),(2,2),(1,2),(0,2),(2,1),(1,1),(0,1),(2,0),(1,0),(0,0)\}.$$
The partial ordering is not a total ordering; for example we neither have $(3,2) \Ppreceq (2,3)$ nor $(3,2) \Psucceq (2,3)$.

\end{example}

An important tool for calculating RGHWs of $q$-ary Reed-Muller codes is
the following corollary to Lemma~\ref{lemfoot}.

\begin{corollary} \label{corfoot} Consider any monomial ordering and let $D =
  \mathrm{span}_{\bbF_q}\{\varphi(F_1),\ldots,\varphi(F_m)\}$ be a
  subspace of $(\bbF_q)^n$ of dimension $m$ where without loss of
  generality we assume $\mathrm{lm}(F_i) = {\vec{X}}^{\vec{a}_{i}} \in R_q^s$
  for $i = 1,\ldots,m$ and $\vec{a}_i \neq \vec{a}_j$ for $i \neq
  j$ (Remark~\ref{rembasisforD}). Writing $A=\{\vec{a}_1, \ldots ,
  \vec{a}_m\}$ we have $|\mathrm{supp}(D)| \geq \left|\nabla A\right|.$
\end{corollary}

\begin{proof}
The elements of $D$ are linear combination of
$\varphi(F_1),\ldots,\varphi(F_m)$, hence $|{\mbox{supp}}(D)|$ equals the length $n$ minus the number of common zeros of $F_1,\ldots,F_m$ over $\bbF_q$. By Lemma \ref{lemfoot} we get
\begin{eqnarray}
&&|\mathrm{supp}(D)| \nonumber \\
&\geq&n-| \Delta_\prec(\langle F_1, \ldots , F_m,X_1^q-X_1, \ldots ,
  X_s^q-X_s\rangle) |\nonumber \\
&\geq&n-\bigg| \bigg( \Delta_\prec(\langle X_1^q-X_1, \ldots ,
  X_s^q-X_s\rangle) \nonumber \\
&&\backslash \cup_{i=1}^m\{\vec{X}^{\vec{a}} \in
  \Delta_{\prec}(\langle X_1^q-X_1, \ldots , X_s^q-X_s\rangle) :
  \vec{X}^{\vec{a}} {\mbox{ is divisible by }} \vec{X}^{\vec{a}_i}\}
  \bigg) \bigg| \nonumber \\
&=&n-|R_q^s|+|\bigcup_{i=1}^m \{\vec{a} \in Q^s_q : \vec{a}
  \Psucceq \vec{a}_{i}\}| = |\bigcup_{i=1}^m \nabla \vec{a}_{i}| =|
  \nabla A|\nonumber
\end{eqnarray}
and the proof is complete.
\end{proof}

Interestingly for any choice of $A$ as in Corollary~\ref{corfoot} there exists some subspaces $D$ for which the bound is sharp.

\begin{proposition} \label{prop1}
Consider any monomial ordering and $A=\{\vec{a}_1, \ldots , \vec{a}_m\} \subseteq Q_q^s$ where
$\vec{a}_i \neq \vec{a}_j$ for $i \neq j$.
Then 
\begin{eqnarray}
\min\{|\mathrm{supp}(D)| :  D =
\mathrm{span}_{\bbF_q}\{\varphi(F_1),\ldots,\varphi(F_m)\} \mbox{
  for some } F_1, \ldots , F_m {\mbox{ \ \ \ \ }}\nonumber \\
{\mbox{ with }}
{\mbox{lm}}(F_i)=\vec{X}^{\vec{a}_i}, i=1, \ldots , m\}=| \nabla A |.\nonumber
\end{eqnarray}
\end{proposition}

\begin{proof}
From Corollary \ref{corfoot} we know that
\begin{eqnarray}
\min\{|\mathrm{supp}(D)| : D =
\mathrm{span}_{\bbF_q}\{\varphi(F_1),\ldots,\varphi(F_m)\} \mbox{
  for some } F_1, \ldots , F_m {\mbox{ \ \ \ \ }}\nonumber \\
{\mbox{ with }}
{\mbox{lm}}(F_i)=\vec{X}^{\vec{a}_i}, i=1, \ldots , m\}\geq | \nabla A |.\nonumber
\end{eqnarray}
Now we want to prove the other inequality.
Let $\bbF_q=\{\gamma_0,\ldots,\gamma_{q-1}\}$ and $\vec{a} = (a_1,\ldots,a_s) \in Q^s_q$, we write $\vec{\gamma}_{\vec{a}} = (\gamma_{a_1},\ldots,\gamma_{a_s})$. For $i=1,\ldots,m$, we write the coordinates of $\vec{a}_{i}$ as $(a_{i,1},a_{i,2},\ldots,a_{i,s})$. We define the following subspace of $(\bbF_q)^n$:
$$\tilde{D} =
\mathrm{span}_{\bbF_q}\left\{\varphi(G_1),\ldots,\varphi(G_m)\right\}\mbox{
  with }G_i = \prod_{t=1}^s\prod_{j=0}^{a_{i,t}-1}(X_t-\gamma_j)\mbox{
  for  }i=1,\ldots,m.$$
For $i=1, \ldots , m$ we have ${\mbox{lm}}(G_i)=\vec{X}^{\vec{a}_i}$. Furthermore $G_i(\gamma_{\vec{a}}) \neq 0$ if and
only if $\vec{a} \in Q_q^s$ satisfies $\vec{a}_{i}
\Ppreceq \vec{a}$. The last result is equivalent to saying that  $G_i(\gamma_{\vec{a}}) \neq 0$ if and
only if $\vec{a} \in \nabla \vec{a}_i$. The support of $\tilde{D}$ is
the union of all positions where some $\varphi(G_i)$ does not equal
$0$. Hence, $|\mathrm{supp}(\tilde{D})| =\left| \bigcup_{i=1}^m \nabla
\vec{a}_{i} \right|=| \nabla A | $. The proof is complete.
\end{proof}

Recall that a q-ary Reed-Muller code is defined as 
$$RM_q(u,s) = {\mbox{span}}_{\mathbb{F}_q}\{\varphi(f) : f \in
R_q^s, \deg(f) \leq u\}.$$
As is well-known~\cite{heijnen1998generalized} the minimum distance strictly increases when
$u$ increases (until the code equals $({\mathbb{F}}_q)^n$). Hence, if we
consider two codes $C_1=RM_q(u_1,s)$, $C_2=RM_q(u_2,s)$ with $u_2 <
u_1$ then 
\begin{equation}
M_1(C_1,C_2)=d \label{eqnabel1}
\end{equation}
where $d$ is the minimum distance of $C_1$. From
Proposition~\ref{prop1} it is not difficult to establish the other
extreme case, namely that of $M_\ell(C_1,C_2)$ where $\ell = \dim C_1
- \dim C_2$. We have $M_\ell(C_1,C_2)=|\nabla A|$ where 
$$A=\{(a_1, \ldots , a_s) : 0\leq a_i < q, i=1, \ldots , s, u_2
< \sum_{i=1}^sa_1 \leq u_1\}.$$
We have $|Q_q^s \backslash \nabla A| =\dim C_2$ and therefore
\begin{equation}
M_\ell (C_1,C_2)=n-\dim C_2. \label{eqnabel2}
\end{equation}
Treating the intermediate cases is much more subtle. This is done in
the following sections.

\section{RGHWs of $q$-ary Reed-Muller codes}\label{secthree}\label{secosaco4}

In this section we employ Proposition \ref{prop1} to compute
the hierarchy of RGHWs in the case that $C_1$ and $C_2$ are both
$q$-ary Reed-Muller codes. The main result is Theorem
\ref{teo4}. 

Our method for calculating the hierarchy of RGHWs involves the
anti lexicographic ordering on the  monomials in $R_q^s$ (and on the
elements in $Q^s_q$). To relate our findings to Heijnen and Pellikaan's
work on GHWs we also need the lexicographic ordering on the same sets.

\begin{definition} \label{defi4}
The lexicographic ordering $\lexprec$ on the monomials in $R_q^s$ and on
the elements in $Q^s_q$ is defined by
$${\vec{X}}^{\vec{a}} \lexprec {\vec{X}}^{\vec{b}}\mbox{(or }\vec{a} \lexprec \vec{b}\mbox{)} \iff a_1=b_1,\ldots,a_{l-1}=b_{l-1} \mbox{ and }a_l < b_l\mbox{ for some }l.$$
The anti lexicographic ordering $\antiprec$ on the monomials in $R_q^s$ and on the elements in $Q^s_q$ is defined by
$${\vec{X}}^{\vec{a}} \antiprec {\vec{X}}^{\vec{b}}\mbox{(or }\vec{a} \antiprec \vec{b}\mbox{)} \iff a_s=b_s,\ldots,a_{s-l+1}=b_{s-l+1} \mbox{ and }a_{s-l} > b_{s-l}\mbox{ for some }l.$$
\end{definition}

\begin{example} \label{ex1}
For $s = 2$, $q=3$ with $X = X_1$ and $Y = X_2$ we have
$$1 \lexprec Y \lexprec Y^2  \lexprec X  \lexprec XY  \lexprec XY^2  \lexprec X^2  \lexprec X^2Y  \lexprec X^2Y^2,$$
$$X^2Y^2 \antiprec XY^2 \antiprec Y^2 \antiprec X^2Y \antiprec XY \antiprec Y \antiprec X^2 \antiprec X \antiprec 1.$$
\end{example}
From this example it is easy to see that the anti lexicographic ordering
is not the inverse ordering of the lexicographic ordering. Recalling from Definition~\ref{defpart} the ordering $\Ppreceq$ we note that if ${\vec{X}}^{\vec{a}} \Ppreceq {\vec{X}}^{\vec{b}}$ (or $\vec{a} \Ppreceq \vec{b}$) then ${\vec{X}}^{\vec{a}} \lexpreceq {\vec{X}}^{\vec{b}}$ and ${\vec{X}}^{\vec{a}} \antisucceq {\vec{X}}^{\vec{b}}$ (or $\vec{a} \lexpreceq \vec{b}$ and $\vec{a} \antisucceq \vec{b}$).

The following concepts will be used extensively throughout our exposition.

\begin{definition}\label{deffw}
Given $\vec{a} = (a_1,\ldots,a_s) \in Q^s_q$, we call $\deg(\vec{a}) = \deg({\vec{X}}^{\vec{a}}) = \sum_{t=1}^s a_t$ the degree of $\vec{a}$. 
Let $a,b$ be two integers with $0 \leq a \leq b \leq s(q-1)$, then we define
$$F_q((a,b),s) =\{\vec{a} \in Q^s_q : a \leq \deg(\vec{a}) \leq b\}\mbox{ and}$$
$$W_q((a,b),s) =\{{\vec{X}}^{\vec{a}} \in R_q^s : \vec{a} \in F_q((a,b),s)\}.$$
\end{definition}

The index $q$ and the value $s$ will be omitted in the rest of this section, thus instead we will use the notations $F(a,b)$ and $W(a,b)$, respectively.

\begin{definition}\label{defnnn}
Let $m \in \{1,\ldots,|F(a,b)|\}$, we denote by $L_{(a,b)}(m)$ the set of the first $m$ elements of $F(a,b)$ using the lexicographic ordering and by $N_{(a,b)}(m)$ the set of the first $m$ elements of $F(a,b)$ using the anti lexicographic ordering.
\end{definition}
The sets $N_{(a,b)}(m)$ will play a crucial role in the following
derivation of a formula for the RGHWs of $q$-ary Reed-Muller
codes. The sets $L_{(a,b)}(m)$ shall help us establish the connection
to the work by Heijnen and Pellikaan on GHWs. Their main
result~\cite[Th.\ 5.10]{heijnen1998generalized} 
is as follows:

\begin{theorem}\label{teoP1}
Let $\vec{a} = (a_1,\ldots,a_s)$ be the $r$-th element in
$F(s(q-1)-u_1,s(q-1))$ with respect to the lexicographic ordering. Then 
\begin{equation}d_r(RM_q(u_1,s)) = |\Delta L_{(s(q-1)-u_1,s(q-1))}(r)| = \sum_{i=1}^s a_{s-i+1}q^{i-1}+1.\label{eqreform}\end{equation}
\end{theorem} 
Before continuing our work on establishing the RGHWs we reformulate
the expressions in~(\ref{eqreform}). We shall need the following
result corresponding to \cite[Lem.\ 5.8]{heijnen1998generalized}.

\begin{lemma}\label{lemP1}
Let $t$ be an integer satisfying $1 \leq t \leq  q^s$. Write $t-1 = \sum_{i=1}^s a_{s-i+1}q^{i-1}$. Then $(a_1,\ldots,a_s)$ is the $t$-th element of $Q^s_q$ with respect to the lexicographic ordering.\end{lemma}

Also we shall need the bijection $\mu: Q^s_q \rightarrow Q^s_q$ given by
$\mu(a_1,\ldots,a_s) = (q-1-a_s,\ldots,q-1-a_1)$. Observe that $\mu$ has the properties 
\begin{itemize}
\item $\vec{a} \antiprec \vec{b} \iff \mu(\vec{a}) \lexprec \mu(\vec{b})$,
\item $\mu(F(a,b)) = F(s(q-1)-b,s(q-1)-a)$,  
\item $\mu(\nabla N_{(a,b)}(m)) = \Delta  L_{(s(q-1)-b,s(q-1)-a)}(m)$.
\end{itemize}
For the proofs and other properties of $\mu$ we refer to
Lemma~\ref{lem3} in Appendix~\ref{A1}. Note that by the first property
 an element $\vec{a}$ in a subset $A$ of $Q^s_q$ is the
$t$-th element in $A$ using the anti lexicographic ordering if and only
if $\mu(\vec{a})$ is the $t$-th element in $\mu(A)$ using the
lexicographic ordering. We can now reformulate Theorem~\ref{teoP1}
into the following result which is not stated in \cite{heijnen1998generalized}. \\ 

\begin{theorem} \label{teo3}
Let $\vec{a}$ be the $r$-th element in $F(0,u_1)$ using the anti
lexicographic ordering. Because $F(0,u_1) \subseteq Q^s_q$ there exists $t$
such that $\vec{a}$ is the $t$-th element in $Q^s_q$ using the anti
lexicographic ordering. We have 
$$d_r(RM_q(u_1,s)) = |\nabla N_{(0,u_1)}(r)| = t.$$
\end{theorem}

\begin{proof}
By the properties of $\mu$ and using the lexicographic ordering, we have that $\mu(\vec{a}) = (\tilde{a}_1,\ldots,\tilde{a}_s)$ is the $r$-th element in $F(s(q-1)-u_1,s(q-1))$ and the $t$-th element in $Q^s_q$. From Theorem \ref{teoP1} we get  
$$d_r(RM_q(u_1,s)) = |\Delta L_{(s(q-1)-u_1,s(q-1))}(r)| = \sum_{i=1}^s \tilde{a}_{s-i+1}q^{i-1} + 1$$
where by Lemma \ref{lemP1} the last expression can be rewritten as $\sum_{i=1}^s \tilde{a}_{s-i+1}q^{i-1} + 1 = t - 1 + 1 = t$.\\
From the third listed property of $\mu$ we obtain $$|\nabla N_{(0,u_1)}(r)| = |\mu(\nabla N_{(0,u_1)}(r))| = |\Delta L_{(s(q-1)-u_1,s(q-1))}(r)|.$$
\end{proof}

Having reformulated the formula by Heijnen and Pellikaan for GHWs we now
continue our work on establishing a formula for the RGHWs. Consider
$C_2 = RM_q(u_2,s) \subsetneq C_1 = RM_q(u_1,s)$. 
Let $\ell$ be the codimension of $C_1$ and $C_2$, then for  $m=1,\ldots,\ell$  
we have 
that
\begin{eqnarray}
M_m(C_1,C_2) & = & \min\{|\mathrm{supp}(D)| : D\mbox{ is a linear
  subcode of }C_1, \nonumber \\
		&&D \cap C_2 = \{\vec{0}\}\mbox{ and }\dim(D) =
m\} \label{eq1}\\
	     & = & \min\{|\mathrm{supp}(D)| : D =
\mathrm{span}_{\bbF_q}\{\varphi(F_1),\ldots,\varphi(F_m)\}, \nonumber \\
	     && \mathrm{lm}(F_1) = {\vec{X}}^{\vec{a}_1}, \ldots,
\mathrm{lm}(F_m) = {\vec{X}}^{\vec{a}_m}, \, \vec{a}_i \neq \vec{a}_j  \mbox{ for $i \neq
  j$}\nonumber \\
		&&{\mbox{and }}	{\vec{X}}^{\vec{a}_{i}} \in W(u_2+1,u_1)\mbox{ for
}i=1,\ldots,m\}
\label{eq2}
\end{eqnarray}
Equation~(\ref{eq1}) corresponds to Theorem
\ref{teo1}. Equation~(\ref{eq2}) follows from
Remark~\ref{rembasisforD} and the fact that $D \subseteq C_1$ implies
${\mbox{lm}}(F_i) \in W(0,u_1)$, $i=1, \ldots , m$ and from the fact
that $D \cap C_2 = \{\vec{0}\}$ implies ${\mbox{lm}}(F_i) \notin
W(0,u_2)$, $i=1, \ldots , m$. In conclusion ${\mbox{lm}}(F_i) \in
W(u_2+1,u_1)$, $i=1, \ldots , m$. Combining~(\ref{eq2}) with
Proposition~\ref{prop1} we get   
\begin{eqnarray}
M_m(C_1,C_2) & = & \min\{|\bigcup_{i=1}^m \nabla \vec{a}_{i}| :
\vec{a}_{i} \in F(u_2+1,u_1), i=1, \ldots m \nonumber \\
&&\mbox{and  } \vec{a}_i \neq \vec{a}_j, {\mbox{ for }} i \neq j\} \nonumber \\
		& = & \min\{|\nabla A| : A \subseteq F(u_2+1,u_1),
|A| = m\}. \label{eq4}
\end{eqnarray}
The following lemma -- which can be viewed as a generalization
of~\cite[Th.\ 3.7.7]{heijnen1999phd} -- is proved in
  Appendix~\ref{A1}.

\begin{lemma} \label{lem5}
Let $A$ be a subset of $F(a,b)$ consisting of $m$ elements. Then $|\nabla N_{(a,b)}(m)| \leq |\nabla A|$.
\end{lemma}

\begin{proposition}\label{rem1}
Let $C_2 = RM_q(u_2,s) \subsetneq C_1 = RM_q(u_1,s)$. We have
\begin{eqnarray*}
M_m(C_1,C_2) & = & |\nabla N_{(u_2+1,u_1)}(m)|
\end{eqnarray*}
\end{proposition}
\begin{proof}
Follows from~(\ref{eq4}) and Lemma \ref{lem5}.
\end{proof}

We are now ready to present the generalization of Theorem~\ref{teo3} to RGHWs.

\begin{theorem} \label{teo4}
Given $C_2 = RM_q(u_2,s) \subsetneq C_1 = RM_q(u_1,s)$, let $\vec{a}$
be the $m$-th element in $F(u_2+1,u_1)$ with respect to the  anti
lexicographic ordering. Because $F(u_2+1,u_1) \subseteq F(0,u_1)
\subseteq Q^s_q$ there exist $r$ and $t$ such that $\vec{a}$ is
the $r$-th element in $F(0,u_1)$ and the $t$-th element in $Q^s_q$ with
respect to the anti lexicographic ordering. We have
$$M_m(C_1,C_2) = t - r + m.$$
\end{theorem}

\begin{proof}
By Proposition~\ref{rem1} we have already proved that $M_m(C_1,C_2) =
|\nabla N_{(u_2+1,u_1)}(m)|$. It remains to be proved that $|\nabla N_{(u_2+1,u_1)}(m)| = t-r+m$.
Because $\vec{a}$ is the $m$-th element in $F(u_2+1,u_1)$ and the $r$-th element in $F(0,u_1)$ we have 
\begin{equation*}
N_{(0,u_1)}(r)  =  N_{(0,u_2)}(r-m)\cup N_{(u_2+1,u_1)}(m) 
\end{equation*}
from which we derive
\begin{eqnarray}
\nabla N_{(0,u_1)}(r) & = & \nabla N_{(u_2+1,u_1)}(m) \cup \nabla
N_{(0,u_2)}(r-m) \nonumber \\
 			& = & \nabla N_{(u_2+1,u_1)}(m) \cup (\nabla
N_{(0,u_2)}(r-m) \backslash \nabla N_{(u_2+1,u_1)}(m)). \label{eqfhs1}
\end{eqnarray}
The union in~(\ref{eqfhs1}) involves two disjoint sets. Hence,
\begin{equation*}
|\nabla N_{(u_2+1,u_1)}(m)|= |\nabla N_{(0,u_1)}(r)|    - |\nabla N_{(0,u_2)}(r-m) \backslash \nabla N_{(u_2+1,u_1)}(m)|.
\end{equation*}
From Theorem~\ref{teo3} we have $|\nabla N_{(0,u_1)}(r)|=t$. Hence, we
will be through if we can prove that 
\begin{equation}
|\nabla N_{(0,u_2)}(r-m) \backslash \nabla N_{(u_2+1,u_1)}(m)|=r-m. \label{eqhjaelp1}
\end{equation}
We enumerate $N_{(0,u_2)}(r-m) = \{\vec{a}_{1},\ldots,\vec{a}_{{r-m}}\}$
according to the anti lexicographic ordering. We have
\begin{eqnarray}
&& \nabla N_{(0,u_2)}(r-m) \backslash \nabla N_{(u_2+1,u_1)}(m) =
  \left(\nabla \bigcup_{i=1}^{r-m} \{\vec{a}_{i}\}\right) \backslash
  \nabla N_{(u_2+1,u_1)}(m) \nonumber \\
& =&  \left(\bigcup_{i=1}^{r-m}  \nabla \vec{a}_{i}\right) \backslash
  \nabla N_{(u_2+1,u_1)}(m) =  \left(\bigcup_{i=1}^{r-m}  \nabla
  \vec{a}_{i} \backslash \nabla \{\vec{a}_{t}: t < i\}\right)
  \backslash \nabla N_{(u_2+1,u_1)}(m) \nonumber \\
&= &\bigcup_{i=1}^{r-m}  \left(\nabla \vec{a}_{i} \backslash
  \left(\nabla \{\vec{a}_{t}: t < i\} \cup \nabla
  N_{(u_2+1,u_1)}(m)\right)\right).\label{eqhjaelp2}
\end{eqnarray}
We will prove that 
\begin{equation}
\nabla \vec{a}_{i} \backslash \left(\nabla
\{\vec{a}_{t}: t < i\} \cup \nabla N_{(u_2+1,u_1)}(m)\right) =
\{\vec{a}_{i}\}\label{eqhjaelp3}
\end{equation}
 holds for $i=1,\ldots,r-m$.\\
As $\vec{a}_{i} \antisucc \vec{a}_{t}$ for $t < i$, we have $\vec{a}_{i} \notin \nabla \{\vec{a}_{t}: t <
i\}$. Furthermore from $\deg(\vec{a}_{i}) \leq u_2$ and
$\deg(\vec{c}) \geq u_2+1$ for any $\vec{c} \in \nabla
N_{(u_2+1,u_1)}(m)$, we conclude $\vec{a}_{i} \notin \nabla
N_{(u_2+1,u_1)}(m)$. It follows that   $$\{\vec{a}_{i}\} \subseteq \nabla \vec{a}_{i} \backslash \left(\nabla \{\vec{a}_{t}: t < i\} \cup \nabla N_{(u_2+1,u_1)}(m)\right).$$
Now we prove the other inclusion. Assume first $\vec{a}_{i} \in
F(u_2,u_2)$. For $t=1,\ldots,s$ we define $\vec{b}_t = \vec{a}_{i} +
\vec{e}_t$ where $\vec{e}_t$ is the standard vector with $1$ in the
$t$-th position. If $\vec{b}_t \in Q^s_q$ then $\vec{b}_t \in
N_{(u_2+1,u_1)}(m)$ because $\deg(\vec{b}_t) = u_2+1$ and $\vec{a} \antisucc \vec{a}_{i}
\antisucc \vec{b}_t$. It follows that $$\nabla \vec{a}_{i} \backslash \left(\nabla \{\vec{a}_{t}: t < i\} \cup \nabla N_{(u_2+1,u_1)}(m)\right) \subseteq \nabla \vec{a}_{i} \backslash \nabla (\{\vec{b}_1,\ldots,\vec{b}_s\}\cap Q^s_q) = \{\vec{a}_{i}\}.$$
Assume next $\vec{a}_{i} \notin F(u_2,u_2)$. Again we define
$\vec{b}_t = \vec{a}_{i} + \vec{e}_t$ for $t=1, \ldots , s$. If
$\vec{b}_t \in Q^s_q$ then $\vec{b}_t \in \{\vec{a}_{t}: t < i\}$
because $\deg(\vec{b}_t) \leq u_2$ and $\vec{a}_{i} \antisucc
\vec{b}_t$. Hence, $$\nabla \vec{a}_{i} \backslash \left(\nabla
\{\vec{a}_{t}: t < i\} \cup \nabla N_{(u_2+1,u_1)}(m)\right)
\subseteq \nabla \vec{a}_{i} \backslash \nabla
(\{\vec{b}_1,\ldots,\vec{b}_s\}\cap Q^s_q) = \{\vec{a}_{i}\}.$$
We have established~(\ref{eqhjaelp3}).\\
Combining finally~(\ref{eqhjaelp3}) and (\ref{eqhjaelp2}) we obtain
$$\nabla N_{(0,u_2)}(r-m) \backslash \nabla N_{(u_2+1,u_1)}(m) =
\bigcup_{i=1}^{r-m} \{\vec{a}_{i}\} = N_{(0,u_2)}(r-m).$$
By Definition~\ref{defnnn} the last set is of size $r-m$ and
(\ref{eqhjaelp1}) follows. The proof is complete.
\end{proof}

Consider the special case of Theorem~\ref{teo4} where
$C_2=\{\vec{0}\}={\mbox{RM}}_q(-1,s)$. In this particular case we
have -- as already noted -- $d_m(C_1)=M_m(C_1,C_2)$. If we apply
Theorem~\ref{teo4} and the notion in there then we obtain $r=m$ and
consequently $M_m(C_1,C_2)=t$. Theorem~\ref{teo3} gives us the same
information $d_m(C_1)=t$. \\

We illustrate the use of Theorem~\ref{teo3} and Theorem~\ref{teo4}
with an example.

\begin{example}  \label{ex5}
In this example we consider Reed-Muller codes in two variables over
${\mathbb{F}}_5$. We first consider the case 
$C_1 = RM_5(5,2)$ and $C_2 = RM_5(3,2)$. Figure~\ref{fig1} 
illustrates how to find $r$ and $m$ for any given $t$ and how to
calculate $d_r(C_1)$ and $M_m(C_1,C_2)$ from this information. The elements
  of $Q^2_5$ are depicted in Part~\ref{fig1}.1. In Parts
  \ref{fig1}.2, \ref{fig1}.3, and \ref{fig1}.4 we illustrate how  the
  elements of $Q^2_5$, $F(0,5)$
and $F(4,5)$, respectively, are ordered. Finally, Part~\ref{fig1}.5
illustrates how to determine $d_r(C_1)$ and $M_m(C_1,C_2)$ from Theorem~\ref{teo3} and Theorem~\ref{teo4}, respectively.

\begin{figure}
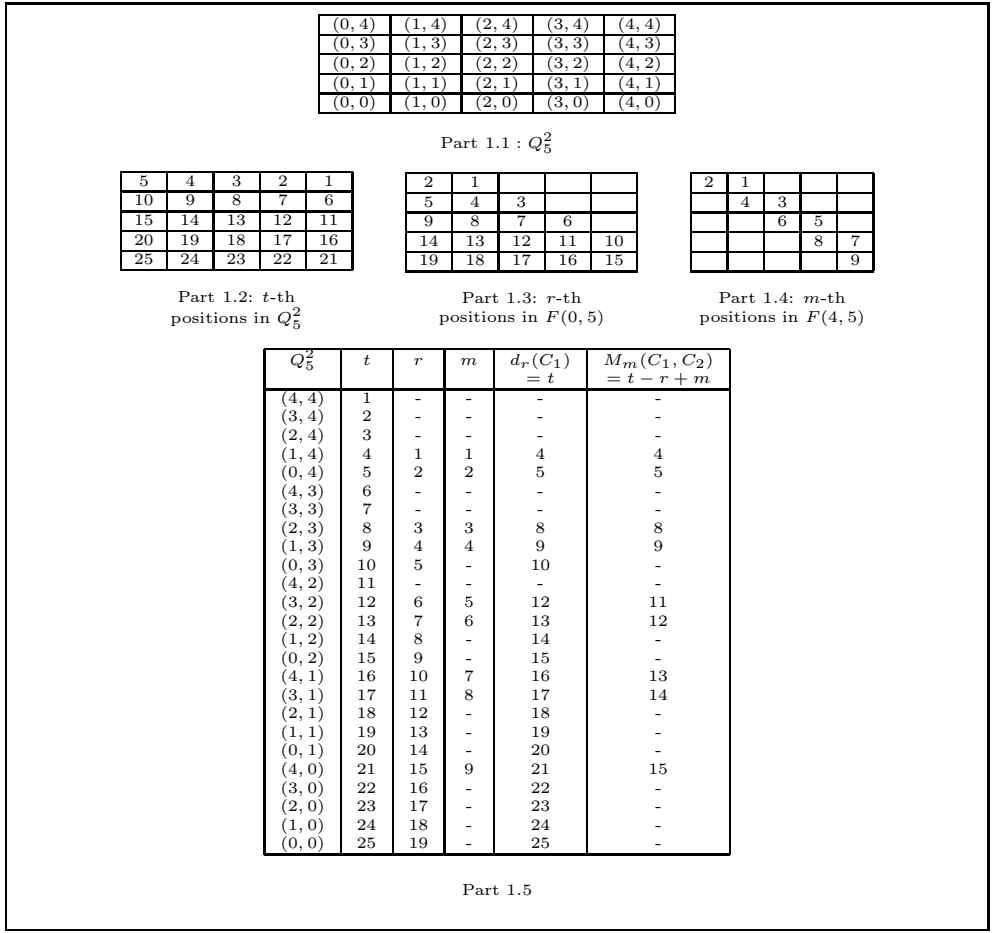
 \tiny
\centering
      \framebox{\parbox{5in}{
		$$	
		\begin{array}{c}
			\begin{array}{c}
			\begin{array}{|c| c| c| c| c|} \hline
				(0,4) & (1,4) & (2,4) & (3,4) & (4,4) \\ \hline
				(0,3) & (1,3) & (2,3) & (3,3) & (4,3) \\ \hline
				(0,2) & (1,2) & (2,2) & (3,2) & (4,2) \\ \hline
				(0,1) & (1,1) & (2,1) & (3,1) & (4,1) \\ \hline
				(0,0) & (1,0) & (2,0) & (3,0) & (4,0) \\ \hline
			\end{array}
 			\\
			\ \\
			\mbox{Part }\ref{fig1}.1: Q^2_5			
			\end{array}

			\\

			\begin{array}{c c c}
			\begin{array}{c}
			\\
			\begin{array}{|c| c| c| c| c|} \hline
				5  &  4 &  3 &  2 & 1  \\ \hline
				10 &  9 &  8 &  7 & 6  \\ \hline
				15 & 14 & 13 & 12 & 11 \\ \hline
				20 & 19 & 18 & 17 & 16 \\ \hline
				25 & 24 & 23 & 22 & 21 \\ \hline
			\end{array}
 			\\
			\ \\
			\mbox{Part }\ref{fig1}.2\mbox{: }t\mbox{-th}\\
			\mbox{positions in }Q^2_5
			\end{array}
			&
			\begin{array}{c}
			\\
			\begin{array}{|c| c| c| c| c|} \hline
				2 & 1 &  &  &  \\ \hline
				5 & 4 & 3 &  &  \\ \hline
				9 & 8 & 7 & 6 &  \\ \hline
				14 & 13 & 12 & 11 & 10 \\ \hline
				19 & 18 & 17 & 16 & 15 \\ \hline
			\end{array}
 			\\
			\ \\
			\mbox{Part }\ref{fig1}.3\mbox{: }r\mbox{-th}\\
			\mbox{positions in }F(0,5)
			\end{array}
			&
			\begin{array}{c}
			\\
			\begin{array}{|c| c| c| c| c|} \hline
				2 & 1 &  &  &  \\ \hline
				 & 4 & 3 &  &  \\ \hline
				 &  & 6 & 5 &  \\ \hline
				 &  &  & 8 & 7 \\ \hline
				 &  &  &  & 9 \\ \hline
			\end{array}
 			\\
			\ \\
			\mbox{Part }\ref{fig1}.4\mbox{: }m\mbox{-th}\\
			\mbox{positions in }F(4,5)
			\end{array}

			\end{array}
			\end{array}
		$$
	 \centering
    \begin{tabular}{ | c | c | c | c | c | c |} 
    \hline
	$Q^2_5$ & $t$ & $r$ & $m$ & $d_r(C_1)$ & $M_m(C_1,C_2)$ \\
	 &  &  &  & $= t$ & $= t - r + m$ \\
    \hline
	$(4,4)$ & 1 & -   &   - &  -	&  - \\
	$(3,4)$ & 2 & -   &   - &  -	&  - \\
	$(2,4)$ & 3 & -   &   - &  -	&  - \\
	$(1,4)$ & 4 & 1   &   1 &  4	&  4 \\
	$(0,4)$ & 5 & 2   &   2 &  5	&  5 \\
	$(4,3)$ & 6 & -   &   - &  -	&  - \\
	$(3,3)$ & 7 & -   &   - &  -	&  - \\
	$(2,3)$ & 8 & 3   &   3 &  8	&  8 \\
	$(1,3)$ & 9 & 4   &   4 &  9	&  9 \\
	$(0,3)$ & 10 & 5   &   - &  10	&  - \\
	$(4,2)$ & 11 & -   &   - &  -	&  - \\
	$(3,2)$ & 12 & 6   &   5 &  12	&  11 \\
	$(2,2)$ & 13 & 7   &   6 &  13	&  12 \\
	$(1,2)$ & 14 & 8   &   - &  14	&  - \\
	$(0,2)$ & 15 & 9   &   - &  15	&  - \\
	$(4,1)$ & 16 & 10   &   7 &  16	&  13 \\
	$(3,1)$ & 17 & 11   &   8 &  17 &  14 \\
	$(2,1)$ & 18 & 12   &   - &  18 &  - \\
	$(1,1)$ & 19 & 13   &   - &  19	&  - \\
	$(0,1)$ & 20 & 14   &   - &  20	&  - \\
	$(4,0)$ & 21 & 15   &   9 &  21 &  15 \\
	$(3,0)$ & 22 & 16   &   - &  22 &  - \\
	$(2,0)$ & 23 & 17   &   - &  23	&  - \\
	$(1,0)$ & 24 & 18   &   - &  24	&  - \\
	$(0,0)$ & 25 & 19   &   - &  25	&  - \\	 
    \hline 
    \end{tabular}
\begin{center}Part \ref{fig1}.5\end{center} 

	}} \caption{Calculation of GHWs and RGHWs for $C_1={\mbox{RM}}_5(5,2)$ and $C_2={\mbox{RM}}_5(3,2)$.}\label{fig1} 
\end{figure}
\clearpage

For the above choice of $C_1$ and $C_2$ most of the time the GHWs and RGHWs are the same. This however, is not the general situation for
$q$-ary Reed-Muller codes as the
following choices of $C_1$ and $C_2$ illustrate.\\
In the remaining part of this example we concentrate on $q$-ary
Reed-Muller codes $C_1={\mbox{RM}}_5(u_1,2)$,
$C_2={\mbox{RM}}_5(u_2,2)$ where $u_1=u_2+1$. In Table~\ref{tabsmal1},
Table~\ref{tabsmal2},  Table~\ref{tabsmal3},  Table~\ref{tabsmal4},
and  Table~\ref{tabsmal5}, respectively, we present parameters
$d_r(C_1)$ and $M_m(C_1,C_2)$ for $(u_1,u_2)$ equal to $(2,1)$,
$(3,2)$, $(4,3)$, $(5,4)$, and $(6,5)$ respectively.   

\begin{table}
\begin{center}
\begin{tabular}{| c| c c |}
\hline
$r=m$ & $d_r(C_1)$ &  $M_m(C_1,C_2)$\\
\hline
1 & 15 & 15 \\
2 & 19 & 19 \\
3 & 20 & 22\\ 
\hline
\end{tabular} 
\end{center}
\caption{$C_1={\mbox{RM}}_5(2,2), C_2={\mbox{RM}}_5(1,2)$.}
\label{tabsmal1}
\end{table}

\begin{table}
\begin{center}
\begin{tabular}{| c| c c |}
\hline
$r=m$ & $d_r(C_1)$ &  $M_m(C_1,C_2)$\\
\hline
1 & 10 & 10 \\
2 & 14 & 14 \\
3 & 15 & 17 \\
4 & 18  & 19 \\
\hline
\end{tabular} 
\end{center}
\caption{$C_1={\mbox{RM}}_5(3,2), C_2={\mbox{RM}}_5(2,2)$.}
\label{tabsmal2}
\end{table}

\begin{table}
\begin{center}
\begin{tabular}{| c| c c |}
\hline
$r=m$ & $d_r(C_1)$ &  $M_m(C_1,C_2)$\\
\hline
1 & 5 & 5 \\
2 & 9 & 9 \\
3 & 10 & 12 \\
4 & 13 & 14 \\
5 & 14 & 15 \\
\hline
\end{tabular} 
\end{center}
\caption{$C_1={\mbox{RM}}_5(4,2), C_2={\mbox{RM}}_5(3,2)$.}
\label{tabsmal3}
\end{table}

\begin{table}
\begin{center}
\begin{tabular}{| c| c c |}
\hline
$r=m$ & $d_r(C_1)$ &  $M_m(C_1,C_2)$\\
\hline
1 & 4 & 4 \\
2 & 5 & 7 \\
3 & 8  & 9  \\
4 & 9  &  10 \\
\hline
\end{tabular} 
\end{center}
\caption{$C_1={\mbox{RM}}_5(5,2), C_2={\mbox{RM}}_5(4,2)$.}
\label{tabsmal4}
\end{table}

\begin{table}
\begin{center}
\begin{tabular}{| c| c c |}
\hline
$r=m$ & $d_r(C_1)$ &  $M_m(C_1,C_2)$\\
\hline
1 & 3 & 3 \\
2 & 4 & 5 \\
3 & 5  & 6  \\
\hline
\end{tabular} 
\end{center}
\caption{$C_1={\mbox{RM}}_5(6,2), C_2={\mbox{RM}}_5(5,2)$.}
\label{tabsmal5}
\end{table}
\end{example}
\newpage
\section{An algorithm to compute RGHWs}\label{secfour}\label{secosaco5}

By Theorem \ref{teo4} there are still two questions that need to be
addressed:
\begin{itemize}
\item[Q1] Given $m \in \{1,\ldots,| F_q((a,b),s)|\}$, how can we find the
$m$-th element $\vec{a}$ of $F_q((a,b),s)$ with respect to the anti lexicographic ordering?
\item[Q2] Given $\vec{a} \in F_q((a,b),s)$ how can we
find 
the corresponding position $t$ and $r$ -- with respect to the anti
lexicographic ordering -- in $Q^s_q$ and in $F_q((0,b),s)$,
respectively?
\end{itemize}
In this section we give answers to these two questions. We start by providing
an algorithm that solves the problem from question Q1. This algorithm
is a generalization of a method proposed in \cite[Sec.\ 6]{heijnen1998generalized}. Due to
the nature of the algorithm from now on we will -- in contrast to the
previous section -- use the full notation
$F_q((a,b),s)$,  rather than just $F_q((a,b))$ (Definition~\ref{deffw}).

\begin{definition} \label{defi7}
Let $0 \leq a \leq b \leq s(q-1)$ and $0 \leq v \leq w < q$ be integers. We define
$$F_q((a,b),(v,w),s) = \{(a_1,\ldots,a_s) \in F_q((a,b),s):  v \leq a_s \leq w\}.$$
We denote by $\rho_q((a,b),s)$ and $\rho_q((a,b),(v,w),s)$ the cardinality of $F_q((a,b),s)$ and $F_q((a,b),(v,w),s)$, respectively. Most of the time the index $q$ will be omitted.
\end{definition}

\begin{figure}
\begin{algorithmic}[1]
\Procedure{VECA}{$A,B,V,S,M,q$: Non-negative integers with $A \leq B
  \leq S(q-1)$,
  $V \leq q-1$, $1 \leq S $, and $M \in \{ 1, \ldots , |F_q((A,B),(0,V),S)| \}$}
\If{$V>B$}
\State ${\mbox{VECA}}(A,B,V,S,M,q) \gets {\mbox{VECA}}(A,B,B,S,M,q)$
\Else 
\If{$S\neq 1$}
\State $\alpha \gets \max \{ A-V,0 \}$
\State $r \gets \rho_q((\alpha,B-V),S-1)$
\If{$M>r$}
\State ${\mbox{VECA}}(A,B,V,S,M,q) \gets {\mbox{VECA}}(A,B,V-1,S,M-r,q)$
\ElsIf{$M<r$}
\State ${\mbox{VECA}}(A,B,V,S,M,q) \gets$
\State \ \ \ \ \ \ $ ({\mbox{VECA}}(\alpha,B-V,q-1,S-1,M,q),V)$
\Else
\State $\theta_1 \gets \alpha {\mbox{ rem }} (q-1)$
\State $\theta_2 \gets (\alpha-\theta_1)/(q-1)$
\If{$\theta_2 < S-1$}
\State ${\mbox{VECA}}(A,B,V,S,M,q)\gets $
\State
\ \ \ \ \ \ $(\underbrace{q-1,\ldots,q-1}_{\theta_2},\theta_1,\underbrace{0,\ldots,0}_{S-\theta_2-2},V)$
\Else
\State ${\mbox{VECA}}(A,B,V,S,M,q)\gets (\underbrace{q-1,\ldots,q-1}_{\theta_2},V)$
\EndIf
\EndIf
\Else
\State ${\mbox{VECA}}(A,B,V,S,M,q) \gets (V-M+1)$
\EndIf
\EndIf
\EndProcedure
\end{algorithmic}
\caption{The recursive algorithm VECA.
We use the
notation $((\beta_1, \ldots , \beta_{\kappa-1}),\beta_\kappa)=(\beta_1, \ldots ,
\beta_{\kappa-1},\beta_\kappa)$ for concatenation.}\label{figveca} 
\end{figure}

\begin{theorem}\label{lem7}
Let $q$ be a fixed prime power and consider  non-negative integers $a,
b, v, s, m$ with 
\begin{equation*}
a \leq b \leq s(q-1), \, v
\leq q-1, \ 1 \leq s, {\mbox{ and }} m \in \{1, \ldots , | F_q((a,b),(0,v),s)|
\}.
\end{equation*}
If these
numbers are used as input to the procedure VECA in
Figure~\ref{figveca} then the output is the $m$-th element
$\vec{a}=(a_1, \ldots , a_s)$ of $F_q((a,b),(0,v),s)$ with respect to
the anti lexicographic ordering.
\end{theorem}
\begin{proof}
Consider the condition
\begin{itemize}
\item[C1:] $A, B,V,S,M$ are non-negative integers with 
$A \leq B \leq S(q-1)$, $V
\leq q-1$, $ 1 \leq S$ and $ M \in \{1, \ldots , | F_q((A,B),(0,V),S)|
\}$.
\end{itemize}
We first show that the following loop invariant holds true:

\begin{itemize}
\item If $V>B$ and $A,B,V,S,M$ satisfy Condition C1 then the elements of  $(\tilde{A}, \tilde{B}, \tilde{V},
  \tilde{S},\tilde{M})=(A,B,B,S,M)$ satisfy Condition C1.
\item If $V \leq B$, $S \neq 1$ and $A,B,V,S,M$ satisfy Condition C1 then:
\begin{itemize}
\item for $M>r$ the elements in  $(\tilde{A}, \tilde{B}, \tilde{V},
  \tilde{S},\tilde{M})=(A,B,V-1,S,M-r)$ satisfy Condition C1,
\item for $M <r$ the elements in $(\tilde{A}, \tilde{B}, \tilde{V},
  \tilde{S},\tilde{M})=(\alpha,B-V,q-1,S-1,M)$ satisfy Condition
  C1. Here $\alpha=\max \{A-V,0\}$.
\end{itemize}
\end{itemize}
Assume first $V >B$. We have $F_q((A,B),(0,V),S)=F_q((A,B),(0,B),S)$
and the result follows. Assume next $V\leq B$ and $S \neq 1$. We consider the case $M>r$ (line 8--9) and leave the case $M<r$ for the
reader. By inspection $\tilde{A} \leq \tilde{B} \leq \tilde{S}(q-1)$,
$\tilde{V} \leq q-1$, $1 \leq \tilde{S}$,  and $\tilde{A},
\tilde{B}, \tilde{S}, \tilde{M}$ are non-negative. Aiming for a
contradiction we assume that $V=0$ is possible (which would cause
$\tilde{V}$ to be negative). But then 
\begin{eqnarray}
r&=&\rho((\alpha,B-V),S-1)=\rho((A,B),S-1)\nonumber \\
&=&\rho((A,B),(0,0),S)=\rho((A,B),(0,V),S)\geq M\nonumber
\end{eqnarray}
where the inequality follows by the assumption that $A, B, V, S, M$
satisfy Condition C1. We have 
reached a contradiction. Hence, we conclude $0 < V$ and therefore $\tilde{V}$ is
non-negative. We next show that $\tilde{M}=M-r$ is in the desired
interval. Clearly $\tilde{M}=M-r \geq 1$. To demonstrate that
$\tilde{M} \leq | F_q((\tilde{A},\tilde{B}),(0,\tilde{V}),\tilde{S})|$
we note that 
\begin{eqnarray}
M &\leq& \rho((A,B),(0,V),S) \nonumber \\
&=&\rho((A,B),(0,V-1),S)+\rho((A,B),(V,V),S)\nonumber \\
&=&\rho((A,B),(0,V-1),S)+\rho((\alpha,B-V),S-1)\nonumber \\
&=&\rho((\tilde{A},\tilde{B}),(0,\tilde{V}),\tilde{S})+r \nonumber
\end{eqnarray}
and the last part of Condition C1 is established.\\
Let $(A_i,B_i,S_i,M_i)$ be the value of $(A,B,S,M)$ before entering
the loop the $i$-th time. The sequence $\big( (A_1, B_1,S_1,M_1),
(A_2, B_2,S_2,M_2), \ldots \big)$ is strictly decreasing with respect
to the partial ordering $\Ppreceq$, and as $A,B,S,M$ are upper bounded as
well as lower bounded the sequence must be finite, meaning that the
algorithm terminates.\\
We next give an induction proof that the algorithm returns the $M$-th
element of $F_q((A,B),(0,V),S)$ with respect to the anti lexicographic
ordering.\\

\noindent {\underline{Basis step:}}\\
First assume $V \leq B$, $S \neq 1$ and let $\theta_1$ and $\theta_2$
be as in line 14 and 15 of the algorithm. Observe that $\theta_2 \leq
S-1$ as $\theta_2=S$ would imply $V=0$ and consequently
$F_q((A,B),(0,V),S)=\emptyset$. This is not possible as
by Condition C1, $M \in \{1, \ldots , | F_q((A,B),(0,V),S) |
\}$. Consider the last element of
$F_q((\alpha,B-V),(V,V),S)$
i.e.\ $$(\underbrace{q-1,\ldots,q-1}_{\theta_2},\theta_1,\underbrace{0,\ldots,0}_{S-\theta_2-2},V)$$
if $\theta_2 < S-1$, and
$$(\underbrace{q-1,\ldots,q-1}_{\theta_2},V)$$
if $\theta_2=S-1$ (in which case $\theta_1=0$). 
This element is the $r$-th element of
$F_q((A,B),(0,V),S)$ where $r$ is as in line 7. Hence, if $M=r$ (lines
13--22) then indeed ${\mbox{VECA}}(A,B,V,S,M,q)$ equals the element in
position $M$ of $F_q((A,B),(0,V),S)$. \\
Assume next $V \leq B$ and $S=1$. We see that the $M$-th element of
$F_q((A,B),(0,V),1)$ equals $(V-(M-1))$ which corresponds
to line 24.\\

\noindent {\underline{Induction step:}}\\
If $V >B$ then as already noted $F_q((A,B),(0,V),S)=F_q((A,B),(0,B),S)$.
For $V \leq B$, $S \neq 1$ we next consider the two cases $M >r$ and $M<r$ separately.\\
We first consider $M>r$ corresponding to lines 8--9 of the
algorithm. We have 
$$\rho((A,B),(V,V),S)=\rho ((\alpha,B-V),(0,q-1),S-1)=r.$$
But $M > r$ and therefore the $M$-th element of $F_q((A,B),(0,V),S)$
equals the $(M-r)$-th element of $F_q((A,B),(0,V-1),S)$.\\
We next consider the case $M < r$. Using similar arguments as above we
see that the $M$-th element of $F_q((A,B),(0,V),S)$ is in
$F_q((A,B),(V,V),S)$. Therefore it equals $(\beta_1, \ldots ,
\beta_{S-1},V)$ where $(\beta_1, \ldots , \beta_{S-1})$ is the
$M$-th element of $F_q((\alpha,B-V),(0,q-1),S-1)$.\\

\noindent The proof is complete.
\end{proof}

Note that for our purpose (that is, to answer Q1), the input $V$ in
the algorithm VECA shall always be equal to $q-1$. The procedure VECA in Figure~\ref{figveca} uses the value
$\rho_q((A,B),S)$ for various choices of $A, B, S$. We
therefore need an algorithm to compute this number.

\begin{lemma}\label{lem120}
Let $q$ be a prime power and consider integers $a,b,s$ with $0 \leq a
\leq b \leq s(q-1)$ and $s \geq 1$. We have
$$\rho_q((a,b),s) = \sum_{i=a}^b \sum_{j=0}^{\lfloor i / q \rfloor} (-1)^j \binom{s}{j} \binom{s-1+i-qj}{s-1}.$$
\end{lemma}
\begin{proof}
We rewrite the first expression as follows
\begin{eqnarray}
\rho_q((a,b),s)& = &|F_q((a,b),s)| = |W_q((a,b),s)| = |W_q((0,b),s)
\backslash W_q((0,a-1),s)|\nonumber \\ 
&=& |W_q(0,b),s)| - |W_q((0,a-1),s)|\nonumber \\
& =& \dim(RM_q(b,s)) -
\dim(RM_q(a-1,s)).\nonumber
\end{eqnarray}
By \cite{sorensen1991projective} and by Exercise 1.2.8 of
\cite{tsfasman1991algebraic} we have that $$\dim(RM_q(u,s)) =
\sum_{i=0}^u \sum_{j=0}^{\lfloor i / q \rfloor} (-1)^j \binom{s}{j}
\binom{s-1+i-qj}{s-1}$$ and the proof follows. 
\end{proof}

\begin{figure}
\begin{algorithmic}[1]
\Procedure{rho}{$a,b,s,q$: Non-negative integers with $0 \leq a \leq
  b\leq s(q-1)$ and $1 \leq s$.}
\State $sum \gets 0$
\For{$i:=a, \ldots , b$}{}
\For{$j:=0, \ldots , \lfloor i/q \rfloor$}{}
\State $sum \gets sum + (-1)^j \binom{s}{j} \binom{s-1+i-qj}{s-1}$
\EndFor
\EndFor
\State \textbf{return} sum
\EndProcedure
\end{algorithmic}
\caption{The algorithm RHO.}\label{figrho}
\end{figure}

\begin{theorem}
Let $q$ be a prime power and consider $a, b, s$ as in Lemma~\ref{lem120}. If the procedure RHO (see
Figure~\ref{figrho}) is used with input $a,b,s,q$ then it returns $\rho_q((a,b),s)$.
\end{theorem}
\begin{proof}
By Lemma~\ref{lem120}.
\end{proof}

Assuming that VECA calls RHO, we can now estimate its time complexity.

\begin{lemma}\label{lemendelig}
The number of binary operations needed to run RHO with input
$a,b,s,q$ is
$${\mathcal{O}}\bigg(\frac{bc}{q} \max \{ s \log q, (s+b)^2 \log^2 (s+b)\}\bigg)$$
where $c=b-a$.
\end{lemma}
\begin{proof}
There are at most $b(b-a)/q$ loop runs in each of which we calculate two binomial
coefficients, perform one multiplication and one addition. According
to~\cite[Example 8, p.\ 7]{koblitz} the number of binary operations
needed to calculate ${{m} \choose {n}}$ is ${\mathcal{O}}
(m^2\log^2n)$. The highest possible $m$ in the algorithm is $N=s+b$
giving at most ${\mathcal{O}}(N^2\log^2N)$ operations for that
task. The multiplication takes place between two numbers no larger than $(s+b)!$
which is ${\mathcal{O}}(N^N)$. As is well-known multiplication of
positive integers $A \geq B$ can be done in
${\mathcal{O}}(\log A\log \log A \log \log \log A)$ binary
operations. In our case this becomes
${\mathcal{O}}(N\log^2N\log\log N)$ which is better than
${\mathcal{O}}(N^2 \log^2N)$. Finally the addition takes place between
numbers equal to $q^s$ at most. Hence, ${\mathcal{O}}(s \log q)$
operations are needed for that part.
\end{proof}
\begin{proposition}
The number of binary operations needed to perform VECA with input $a,
b, v, s, M, q$ is
$${\mathcal{O}}\bigg(sb^2\max\{ s \log q, (s+b)^2\log^2(s+b)\}\bigg).$$
\end{proposition}
\begin{proof}
If the output of VECA is $(q-g_1, \ldots ,q-g_s)$ then in the worst
case RHO is called $g_1+ \cdots +g_s$ times. Hence, in the worst case
VECA calls RHO $sq$ times. In each call the first input of RHO is lower
bounded by $0$, the second is upper bounded by $b$, and the third is
upper bounded by $s$. The result now follows from Lemma~\ref{lemendelig}.
\end{proof}

\begin{example} \label{exAAA}
We use the algorithm VECA in Figure~\ref{figveca} to find the $34$-th
element $\vec{a}=(a_1, \ldots , a_7)$ of $F_7((20,22),7)$. The
procedure takes as input $(A,B,V,S,M) = (20,22,6,7,34)$. The notation
$\tilde{A}$, $\tilde{B}$, $\tilde{V}$, $\tilde{S}$, $\tilde{M}$ is as
in the proof of Theorem~\ref{lem7}.\\

\noindent {\underline{$(A,B,V,S,M)=
(20,22,6,7,34)$:}}\\
$\rho_7((14,16),6) = 23415 > 34$ (lines 10--12). Thus $a_7 = 6$,
$\tilde{A} = \max\{0,20-6\} = 14$, $\tilde{B} = 22 - 6 = 16$,
$\tilde{V} = q-1 = 6$ and $\tilde{S} = 7-1 = 6$.\\

\noindent {\underline{$(A,B,V,S,M)=
(14,16,6,6,34)$:}}\\
$\rho_7((8,10),5) = 1936 > 34$ (lines 10--12). Thus $a_6 = 6$, $\tilde{A}
= \max\{0,14-6\} = 8$, $\tilde{B} = 16 - 6 = 10$, $\tilde{V} = q-1 =
6$ and $\tilde{S} = 6-1 = 5$.\\

\noindent {\underline{$(A,B,V,S,M)= 
(8,10,6,5,34)$:}}\\ 
$\rho_7((2,4),4) = 64 > 34$ (lines 10--12). Thus $a_5 = 6$, $\tilde{A} =
\max\{0,8-6\} = 2$, $\tilde{B} = 10 - 6 = 4$, $\tilde{V} = q-1 = 6$
and $\tilde{S} = 5-1 = 4$.\\

\noindent {\underline{$(A,B,V,S,M)= 
(2,4,6,4,34)$:}}\\
$6 > 4$ (lines 2--3). Thus $\tilde{V} = B = 4$.\\

\noindent {\underline{$(A,B,V,S,M)=
(2,4,4,4,34)$:}}\\
$\rho_7((0,0),3) = 1 < 34$ (lines 8--9). Thus $\tilde{M} = 34 - 1 = 33$
and $\tilde{V} = 4-1 = 3$.\\

\noindent {\underline{$(A,B,V,S,M)=
(2,4,3,4,33)$:}}\\ 
$\rho_7((0,1),3) = 4 < 33$ (lines 8--9). Thus $\tilde{M} = 33 - 4 = 29$
and $\tilde{V} = 3-1 = 2$.\\

\noindent {\underline{$(A,B,V,S,M)=
(2,4,2,4,29)$:}}\\
$\rho_7((0,2),3) = 10 < 29$ (lines 8--9). Thus $\tilde{M} = 29 - 10 = 19$
and $\tilde{V} = 2-1 = 1$.\\

\noindent {\underline{$(A,B,V,S,M)=
(2,4,1,4,19)$:}}\\
$\rho_7((1,3),3) = 19 = 19$ (lines 13--17). We have $\theta_1 = 1$ and $\theta_2 = 0$, thus $(a_1,a_2,a_3,a_4) = (1,0,0,1)$ and the algorithm ends.\\

\noindent In conclusion the $34$-th element of $F_7((20,22),7)$ is $(a_1,a_2,a_3,a_4,a_5,a_6,a_7) = (1,0,0,1,6,6,6)$.
\end{example}

Having answered question Q1 from the beginning of the section we
now turn to question Q2. Given $\vec{a} \in F_q((a,b),s)$ we need a method to
determine what are the corresponding positions $r$ and $t$ in $F_q((0,b),s)$ and $Q^s_q$, respectively. The following proposition tells us how to find
$r$. This is done by applying the formula~(\ref{hejduder}) in there in
combination with the algorithm RHO.

\begin{proposition} \label{lem8}
The element $\vec{a} = (a_1,\ldots,a_s) \in F_q((a,b),s)$ is the $r$-th
element of $F_q((a,b),s)$ with respect to the anti lexicographic ordering, where 
$$r = \sum_{j=0}^{s-1} \sum_{i=0}^{q-a_{s-j}-2}
\rho_q((\max\{0,a-\sum_{t=0}^{j}a_{s-t}-i-1\},b-\sum_{t=1}^{j}a_{s-t}-i-1),s-j-1)
+ 1.$$
In particular if $a=0$ then
\begin{equation}
r = \sum_{j=0}^{s-1} \sum_{i=0}^{q-a_{s-j}-2}
\rho_q((0,b-\sum_{t=1}^{j}a_{s-t}-i-1),s-j-1) + 1. \label{hejduder}
\end{equation}
\end{proposition}

\begin{proof}
We must count the number of elements $\vec{b}=(b_1, \ldots , b_s)$ in
$F_q((a,b),s)$ which are smaller than or equal to $\vec{a}$ with respect to
the anti lexicographic ordering. This number equals
\begin{eqnarray}
r& =& |\{\vec{b} \in F_q((a,b),s) : \vec{b} \antipreceq \vec{a}\}|
\nonumber \\
&=& |\{\vec{b} \in F_q((a,b),s) : b_s > a_s\}| + |\{\vec{b} \in F_q((a,b),s)
: \vec{a} \antipreceq \vec{b}, b_s = a_s\}|\nonumber \\
&=& \rho((a,b),(a_s+1,q-1),s) + |\{\vec{b} \in F_q((a,b),s) : b_{s-1} >
a_{s-1}, b_s = a_s\}| \nonumber \\
&&+ |\{\vec{b} \in F_q((a,b),s) : \vec{a} \antipreceq \vec{b}, b_{s-1} =
a_{s-1}, b_s = a_s\}| \nonumber \\
&=&  \rho((a,b),(a_s+1,q-1),s) +
\rho((\max\{0,a-a_s\},b-a_s),(a_{s-1}+1,q-1),s-1) \nonumber \\
&&+ |\{\vec{b} \in F_q((a,b),s) : \vec{a} \antipreceq \vec{b}, b_{s-1} =
a_{s-1}, b_s = a_s\}| \nonumber \\
&=&\cdots  \nonumber \\
&=& \sum_{j=0}^{s-1}
\rho((\max\{0,a-\sum_{t=0}^{j-1}a_{s-t}\},b-\sum_{t=0}^{j-1}a_{s-t}),(a_{s-j}+1,q-1),s-j)
+ |\{\vec{a}\}| \nonumber \\
&= &\sum_{j=0}^{s-1} \rho((\max\{0,a-\sum_{t=0}^{j-1}a_{s-t}\},b-\sum_{t=0}^{j-1}a_{s-t}),(a_{s-j}+1,q-1),s-j) + 1.\nonumber
\end{eqnarray}
By the below Lemma \ref{lemsyv}, for $j = 0,\ldots,s-1$ we have
\begin{eqnarray}
&&\rho((\max\{0,a-\sum_{t=0}^{j-1}a_{s-t}\},b-\sum_{t=0}^{j-1}a_{s-t}),(a_{s-j}+1,q-1),s-j)\nonumber \\
&=& \sum_{i=0}^{q-a_{s-j}-2}
  \rho((\max\{0,a-\sum_{t=0}^{j}a_{s-t}-i-1\},b-\sum_{t=1}^{j}a_{s-t}-i-1),s-j-1)\nonumber 
\end{eqnarray}
and the proof is complete.
\end{proof}

\begin{lemma} \label{lemsyv}
Given a prime power $q$, let $0 \leq a \leq b \leq s(q-1)$  and $0 \leq v \leq w < \min\{b,q\}$
be integers. Then $\rho_q((a,b),(v,w),s) = \sum_{i=0}^{w-v} \rho_q((\max\{0,a-v-i\},b-v-i),s-1)$.
\end{lemma}

\begin{proof}
\begin{eqnarray}
\rho((a,b),(v,w),s) &= &|F_q((a,b),(v,w),s)| \nonumber \\
&=&  |\{(a_1,\ldots,a_s) \in F_q((a,b),s):  v \leq a_s \leq w\}|
\nonumber \\
&=&|\bigcup_{i=0}^{w-v} \{(a_1,\ldots,a_s) \in F_q((a,b),s):  a_s =
v+i\}|\nonumber \\
& =& |\bigcup_{i=0}^{w-v} F_q((a,b),(v+i,v+i),s)| \nonumber \\
%&=&\sum_{i=0}^{w-v} |F_q((a,b),(v+i,v+i),s)|\nonumber \\
& =& \sum_{i=0}^{w-v} \rho((a,b),(v+i,v+i),s)\nonumber \\
& = &\sum_{i=0}^{w-v} \rho((\max\{0,a-v-i\},b-v-i),s-1).\nonumber 
\end{eqnarray}
\end{proof}

Setting $a=0$ and $b=s(q-1)$ in Proposition~\ref{lem8} we could of
course compute the $t$ such that $\vec{a}$ is the $t$-th element of
$Q_q^s$, but with the following reformulation of Lemma~\ref{lem8} we
can calculate it much easier.

\begin{lemma} \label{lem9}
The element $(a_1,\ldots,a_s) \in Q^s_q$ is the $t$-th element of $Q^s_q$
with respect to the anti lexicographic ordering where 
$$t = q^s - \sum_{i=1}^s a_i q^{i-1}.$$
\end{lemma}

\begin{proof}
Recall from Section~\ref{secthree} the map $\mu: Q^s_q \rightarrow Q^s_q$, 
$\mu(a_1,\ldots,a_s) = (q-1-a_s,\ldots,q-1-a_1)$. By Lemma \ref{lemP1}
$\mu(a_1,\ldots,a_s) = (q-1-a_s,\ldots,q-1-a_1)$ is the $t$ element of
$Q^s_q$ using the lexicographic ordering where
$t-1 = \sum_{i=1}^s (q-1-a_i)q^{i-1} = q^s - 1 - \sum_{i=1}^s a_i
q^{i-1}$. Recall from Section~\ref{secthree} that
$\vec{c} \antiprec \vec{d} \iff \mu(\vec{c}) \lexprec \mu(\vec{d})$. 
Therefore $(a_1,\ldots,a_s)$ is the $t$-th element of $Q^s_q$ using the anti lexicographic ordering.
\end{proof}

Summarizing this section: to find the $m$-th RGHW of $C_1 =
RM_q(u_1,s)$ with respect to $C_2 = RM_q(u_2,s)$, we perform the following steps.
\begin{enumerate}
\item Find the $m$-th element $(a_1,\dots,a_s)$ of
  $F_q((u_2+1,u_1),s)$ by using the algorithm VECA in
  Theorem~\ref{lem7} with input $A=u_2+1$, $B=u_1$, $V=q-1$, $S=s$,
  and $M=m$.
\item Find the $r$-th position of $(a_1,\dots,a_s)$ in
  $F_q((0,u_1),s)$ using Proposition~\ref{lem8} in combination with
  the algorithm RHO.
\item Find the $t$-th position of $(a_1,\dots,a_s)$ in $Q^s_q$
  using Lemma~\ref{lem9}.
\item Compute $M_m(C_1,C_2) = t - r + m$ (Theorem \ref{teo4}).
\end{enumerate}

\begin{example} \label{ex5B}
This is a continuation of Example~\ref{ex5}, in the beginning of which
we considered
$C_1 = RM_5(5,2)$ and $C_2 = RM_5(3,2)$. Applying the above procedure
to establish the $8$-th RGHW we first use Theorem~\ref{lem7} to
establish that the $8$-th element of $F_5((4,5),2)$ is $(3,1)$.
Using Proposition \ref{lem8} we then find that $(3,1)$ is the $11$-th
element of $F_5((0,5),2)$ and Lemma~\ref{lem9} next tells us that it is the $17$-th element of $Q^2_5$. Hence, $M_8(C_1,C_2) = 17-11+8 = 14$.
\end{example}

\begin{example} \label{ex5C}
We consider $C_1 = RM_{16}(90,7)$ and $C_2 = RM_{16}(88,7)$. We want
to compute the $1000$-th RGHW  of $C_1$ with respect to $C_2$. Applying the algorithm VECA in Theorem  \ref{lem7} we find that that
$(9,10,14,11,15,15,15)$ is the $1000$-th element of $F_{16}((88,90),7)$. Applying next
Proposition~\ref{lem8} and Lemma~\ref{lem9} we find that it is the $14557$-th
element of $F_{16}((0,90),7)$ and the $16727$-th element of $Q^7_{16}$. Hence,
$M_{1000}(C_1,C_2) = 16727-14557+1000 = 3170$. To find the $1000$-th
GHW of $C_1$, we use Theorem~\ref{lem7} with $C_2 = RM_{16}(-1,7)$ and
we find that $(5,1,10,15,15,15,15)$ is the $1000$-th element of
$F_{16}((0,90),7)$. By Lemma~\ref{lem9} it is the $1515$-th element of
$Q^7_{16}$. Hence, from Theorem \ref{teo3} we deduce $d_{1000}(C_1) = 1515$.
\end{example}

\section{Closed formula expressions for $q$-ary Reed-Muller codes in
  two variables} \label{RM2}\label{secosaco6}

In the previous section we presented a method to calculate RGHWs for
any set of $q$-ary Reed-Muller codes $C_i={\mbox{RM}}_q(u_i,s)$,
$i=1, 2$. As an alternative, for $q$-ary Reed-Muller codes in two
variables (which by Definition~\ref{defi0} means that $s=2$) it is a
manageable task to list closed formula expressions for all possible
situations. This is done in the first half of the present
section. Letting next $u_2=-1$, corresponding to $C_2=\{ \vec{0}\}$,
we in particular 
get closed formula expressions for the GHWs (such formulas -- to the
best of our knowledge -- cannot be found in the literature). 
The formulas in the present section can be derived
by applying Proposition~\ref{rem1} directly. We shall leave the
details for the reader. To simplify the description we use the
notation $t=u_1-u_2$ which of course implies that
$u_1=u_2+t$. Hence, throughout this section $C_2={\mbox{RM}}_q(u_2,2)$
and $C_1={\mbox{RM}}_q(u_2+t,2)$.

\subsection{Formulas for RGHW}  \label{RGHWRM2}

We have the following three cases.

\subsubsection{First case: $u_2-q+2 \geq 0$}

\begin{figure}[!thpb] \label{figure11}
\centering
       $$
		\begin{array}{c} \\
			\begin{array}{c c c c c}
				Y^4 	& \underline{XY^4}	& \underline{X^2Y^4}	& X^3Y^4		& X^4Y^4		\\
				Y^3 	& XY^3 			& \underline{X^2Y^3} 	& \underline{X^3Y^3} 	& X^4Y^3		\\
				Y^2	& XY^2 			& X^2Y^2 		& \underline{X^3Y^2} 	& \underline{X^4Y^2}	\\
				Y 	& XY 			& X^2Y   		& X^3Y 			& \underline{X^4Y} 	\\
				1 	& X 			& X^2    		& X^3 			& X^4 	
			\end{array}
 			\\
			\ \\
			W_5(5,6)\mbox{ underlined, i.e. }u_2=4\mbox{ and }t=2\\
			\mbox{(First case)}
		\end{array}
	$$ 
\end{figure}

In this case the codimension is $\ell = t(2q-u_2-t-2)+\frac{t(t+1)}{2}$.

\begin{itemize}
\item If $m = 1,\ldots,t(2q-u_2-t-2)$ then there exist $a \in
  \{0,\ldots,2(q-1)-u_2-t-1\}$ and $b\in\{1,\ldots,t\}$ such that $m =
  at+b$. We have
$$M_m(C_1,C_2) = \left(2q-2-u_2-\frac{a}{2}\right)(a+1) + b - t.$$

\item If $m = t(2q-u_2-t-2)+1,\ldots,t(2q-u_2-t-2)+\frac{t(t+1)}{2}$,
  then there exists $c \in \left\{1,\ldots,\frac{t(t+1)}{2}\right\}$
  such that $m = t(2q-u_2-t-2) + c$. We have
$$M_m(C_1,C_2) = \frac{1}{2}(2q-u_2-t-2)(2q-u_2+t-1)+c.$$
\end{itemize}

\subsubsection{Second case: $u_2-q+t+1 \leq 0$}

\begin{figure}[!thpb] \label{figure22}
\centering
     $$
			\begin{array}{c}
			\\
			\begin{array}{c c c c c}
				Y^4 		& XY^4		& X^2Y^4	& X^3Y^4	& X^4Y^4	\\
				\underline{Y^3} 	& XY^3 		& X^2Y^3 	& X^3Y^3 	& X^4Y^3	\\
				\underline{Y^2}	& \underline{XY^2} 	& X^2Y^2 	& X^3Y^2 	& X^4Y^2	\\
				Y 		& \underline{XY} 	& \underline{X^2Y}   	& X^3Y 		& X^4Y 		\\
				1 		& X 		& \underline{X^2}    	& \underline{X^3} 	& X^4 	
			\end{array}
 			\\
			\ \\
			W_5(2,3)\mbox{ underlined, i.e. }u_2=1\mbox{ and }t=2\\
			\mbox{(Second case)}
			\end{array}
$$
\end{figure}

In this case the codimension is $\ell =\frac{t(t+1)}{2}+t(u_2+1)$.

\begin{itemize}
\item If $m = 1,\ldots,\frac{t(t+1)}{2}$ then there exist $a \in
  \{0,\ldots,t-1\}$ and $b\in\{1,\ldots,a+1\}$ such that $m =
  \frac{a(a+1)}{2}+b$. We have
$$M_m(C_1,C_2) = q(q-u_2-t+a)+b-a-1.$$

\item If $m = \frac{t(t+1)}{2}+1,\ldots,\frac{t(t+1)}{2}+t(u_2+1)$,
  then there exist $a \in \{0,\ldots,u_2\}$ and $b\in\{1,\ldots,t\}$
  such that $m = \frac{t(t+1)}{2} + at +b$. We have
$$M_m(C_1,C_2) = q(q+a-u_2)+b-t-1-\frac{a(a+3)}{2}.$$
\end{itemize}

\subsubsection{Third case: $u_2-q+2 < 0$ and $u_2-q+t+1 > 0$}

\begin{figure}[!thpb] \label{figure33}
\centering
      $$
			\begin{array}{c}
			\\
			\begin{array}{c c c c c}
				\underline{Y^4} 	& \underline{XY^4}	& X^2Y^4	& X^3Y^4		& X^4Y^4	\\
				\underline{Y^3} 	& \underline{XY^3} 	& \underline{X^2Y^3} 	& X^3Y^3 		& X^4Y^3	\\
				Y^2		& \underline{XY^2} 	& \underline{X^2Y^2} 	& \underline{X^3Y^2} 		& X^4Y^2	\\
				Y 		& XY 		& \underline{X^2Y}   	& \underline{X^3Y} 		& \underline{X^4Y} 	\\
				1 		& X 		& X^2    	& \underline{X^3}		& \underline{X^4} 	
			\end{array}
 			\\
			\ \\
			W_5(3,5)\mbox{ underlined, i.e. }u_2=2\mbox{ and }t=3\\
			\mbox{(Third case)}
			\end{array}
$$
\end{figure}

In this case the codimension is $\ell = (2q-u_2)(u_2+t)+3(q-u_2)-q^2-2-\frac{t(t+3)}{2}$.

\begin{itemize}
\item If $m = 1,\ldots,\frac{1}{2}(q-u_2-2)(2t-q+u_2+1)+t$ then there
  exist $a \in \{0,\ldots,q-u_2-2\}$ and
  $b\in\{1,\ldots,u_2+t-q+a+2\}$ such that $m =
  a(u_2+t-q+1)+\frac{a(a+1)}{2}+b$. We have 
$$M_m(C_1,C_2) = (a+2)(q-1)-u_2-t+b.$$

\item If $m =
  \frac{1}{2}(q-u_2-2)(2t-q+u_2+1)+t+1,\ldots,\frac{1}{2}(q-u_2-2)(2t-q+u_2+1)+t(q-t)$
  then there exist $a \in \{0,\ldots,q-t-2\}$ and $b\in\{1,\ldots,t\}$
  such that $m = \frac{1}{2}(q-u_2-2)(2t-q+u_2+1)+(a+1)t+b$. We have
$$M_m(C_1,C_2) = q(q-u_2+a)-\frac{a(a+3)}{2}-t+b-1.$$

\item If $m =
  \frac{1}{2}(q-u_2-2)(2t-q+u_2+1)+t(q-t)+1,\ldots,(2q-u_2)(u_2+t)+3(q-u_2)-q^2-2-\frac{t(t+3)}{2}$
  then there exists $c \in
  \{1,\ldots,\frac{1}{2}((t+1)^2-(q-u_2-1)^2+q-u_2-t-2)\}$ such that
  $m = \frac{1}{2}(q-u_2-2)(2t-q+u_2+1)+t(q-t)+c$. We have
$$M_m(C_1,C_2) = \frac{1}{2}(3q^2-2u_2q-3q-t^2-t)+c.$$
\end{itemize}

%It is interesting to note what will happen when we compare $M_j(C_1,C_2)$ with $n-M_{l-j+1}({C_2}^\perp,{C_1}^\perp)+1$.\\
%Let $C_1 = RM_q(s+1,2)$ and $C_2 = RM_q(s,2)$, then by \cite{peterson1961error}, we obtain ${C_1}^\perp = RM_q(2q-s-3)$, ${C_2}^\perp = RM_q(2q-s-2)$. Note that their codimension is still $l$.
%Using (1), for all $j = 1,\ldots,l$ we obtain that:
%$$(n-M_{l-j+1}({C_2}^\perp,{C_1}^\perp)+1) - M_j(C_1,C_2) = j^2 + aj + b \mbox{ where }$$
%$$a = -\min\{q,2q-s-2\} -\min\{s+2,2q-s-2\}+\min\{q,s+2\}-1\mbox{ and }$$
%$$b = \frac{1}{2}(\min\{s+2,2q-s-2\})^2 - \min\{s+2,2q-s-2\}\min\{s+2,q\} +$$ $$ \frac{1}{2}\min\{s+2,2q-s-2\} - \min\{q,s+2\} - q\|q-s-2\| + q^2 + 1.$$
%Because $s,q$ are fixed we have that $a$ and $b$ are constant. In relation to $j$, $j^2 + aj + b$ is a convex parabola. Thus the minimum is in 
%$$j_{\mathrm{min}} = \begin{cases}
%1 & \mbox{ if }\frac{a}{2} < 1\\
%l & \mbox{ if }\frac{a}{2} > l\\
%\frac{a}{2} & \mbox{ otherwise}
%\end{cases}.$$

\subsection{Formulas for GHW}

Applying the formulas from the previous section to the special case
of $u_2=-1$ and consequently $u_1=t+1$ we get by letting $u=u_1$ the following results
concerning the GHWs of ${\mbox{RM}}_q(u,s)$.

\subsubsection{The case $u-q+1\leq 0$}
In this case the dimension of $C_1$ is $k_1 = \frac{(u+1)(u+2)}{2}$.
\begin{itemize}
\item For $r = 1,\ldots,\frac{(u+1)(u+2)}{2}$ there exist $a
  \in \{0,\ldots,u\}$ and $b\in\{1,\ldots,a+1\}$ such that $r =
  \frac{a(a+1)}{2}+b$. We have
$$d_r(C_1) = q(q-u+a)+b-a-1.$$
\end{itemize}

\subsubsection{The case $u-q+1 > 0$}
In this case the dimension of $C_1$ is $k_1 = q(2u_1-q+3)-\frac{u_1(u_1+3)}{2}-1$.
\begin{itemize}
\item For $r = 1,\ldots,q(u+2)-\frac{u(u+3)}{2}-1$ there
  exist $a \in \{0,\ldots,2(q-1)-u\}$ and
  $b\in\{1,\ldots,u-q+2+a\}$ such that $r =
  a(u-q+1)+\frac{a(a+1)}{2}+b$. We have
$$d_r(C_1) = (a+2)(q-1)-u+b.$$

\item For $r =
  q(u+2)-\frac{u(u+3)}{2},\ldots,q(2u-q+3)-\frac{u(u+3)}{2}-1$
  there exists $c \in \{1,\ldots,q(u-q+1)\}$ such that $r =
  q(u+2)-\frac{u(u+3)}{2}-1+c$. We have
$$d_r(C_1) = q(2q-u-1)+c.$$
\end{itemize}

\subsection{Comparing RGHW and GHW in a special case}

Consider the special case $u_2 = q-2$ and $t=1$. 
If $m = 1,\ldots,q$ then there exist $a \in \{0,\ldots,q-1\}$ and
$b\in\{1,\ldots,a+1\}$ such that $m = \frac{a(a+1)}{2}+b$. We have
$$M_m(C_1,C_2) = \frac{m}{2}(2q-m+1) \mbox{ and }d_m(C_1) = (q-1)(a+1)+b$$
Thus 
$$M_m(C_1,C_2) - d_m(C_1) = \frac{1}{8}(-a^4-2a^3+(-4b+4q+1)a^2$$
$$+(-4b-4q+10)a-4b^2+8bq-4b-8q+8).$$

For the particular case that $q=16$ we get the values listed in
Table~\ref{tabslut}

\begin{table}
\begin{center}
\begin{tabular}{|c|rrrrrrrrrr|}
\hline
m&1&2&3&4&5&6&7&8&9&10\\
\hline
${\mbox{diff}}(m)$&0&0&14&15&29&43&45&59&73&87\\
$M_m(C_1,C_2)$&16&31&46&61&76&91&106&121&136&151\\
\hline
\end{tabular}
\end{center}
\begin{center}
\begin{tabular}{|c|rrrrrr|}
\hline
m&11&12&13&14&15&16\\
\hline
${\mbox{diff}}(m)$&90&104&118&132&146&150\\
$M_m(C_1,C_2)$&166&181&196&211&226&241\\
\hline
\end{tabular}
\end{center}
\caption{The special case $u_2=q-2$ and $t=1$ with $q=16$. That is,
  $C_1={\mbox{RM}}_{16}(15,2)$ and $C_2={\mbox{RM}}_{16}(14,2)$. The
  function ${\mbox{diff}}(m)$ equals $M_m(C_1,C_2)-d_m(C_1)$.}
\label{tabslut}
\end{table}

\section{Locally correctable ramp secret sharing schemes}\label{secosaco7}

We now return to the communication problem described in the
introduction of the paper. Recall that we consider a secret sharing
scheme based on a coset construction $C_1/C_2$ where $C_1$ and $C_2$
are $q$-ary Reed-Muller codes. Requirement R2  about local
correctability was treated in Theorem~\ref{thelocal}. In
Section~\ref{sec3} -- Section~\ref{RM2} we showed a low complexity method
to determine the RGHWs and in particular we derived closed formula
expressions in the case of codes in two variables. By the following
result (corresponding to~(\ref{eqsnabeleins}) and (\ref{eqsnabeltwei})) 
\begin{eqnarray}
t_m=M_m((C_2)^\perp,(C_1)^\perp)-1, & r_m=n-M_{\ell-m+1}(C_1,C_2)+1,  \label{eqny2}
\end{eqnarray}
this method immediately translates into accurate
information on the information leakage and thereby explains what can be
done regarding requirement R1.\\

Combining (\ref{eqny2}) with (\ref{eqnabel1}) and
(\ref{eqnabel2}) and using Remark~\ref{remduality} we obtain
\begin{equation}
\begin{array}{ll}
t_1=d(C_2^\perp)-1, & r_1=\dim (C_2)+1,\\
t_\ell=\dim (C_1) -1,& r_\ell = n-d(C_1)+1,
\end{array} \label{eqhejsadu}
\end{equation}
where $d(C)$ is the minimum distance of $C$. To apply
Theorem~\ref{thelocal} (which ensures local correctability) we
need $u_1 < q-1$. Under that assumption (\ref{eqhejsadu}) becomes
\begin{equation}
\begin{array}{ll}
t_1=u_2+1, & r_1={{s+u_2} \choose {u_2}} +1,\\
t_\ell={{s+u_1} \choose {u_1}} -1,& r_\ell = q^s-(q-u_1)q^{s-1}+1=u_1 q^{s-1}+1.
\end{array} \label{eqhejsathere}
\end{equation}
By Theorem~\ref{thelocal} we need to make $u_1+1$ or $q-1$ queries
(depending on the error-probability of the system) to hopefully correct an
entry. We observe that the number of queries in both cases is 
strictly larger than $t_1$. However, it is only larger than
$r_1$ when $u_2$ is very small. Actually, for most values of $u_2$ the
number of queries needed will be much smaller than $r_1$. Recall from
the proofs in~\cite{now} of the local correctability of $q$-ary
Reed-Muller codes that the random point sets queried is
chosen from a family of point sets with a particular geometry (the
geometry is different for the three different cases treated in
Theorem~\ref{thelocal}). Knowing only the values $t_1$ and $r_1$ -- with
the number of queries being a number in between -- we
cannot say if those point sets get access to information or
not. However, when $t_2$ is larger than or equal to the number of queries then
for sure they get at most access to 1 $q$-bit of information. As is
demonstrated in the following two examples this is often the case. Of
course the
situation gets more complicated if the decoding is not
successful in the first run and another series of queries is needed. In that case we
may either ensure that the information from the first query is
deleted or we may simply trust the party that performs the
error-correction. Below we study in detail various schemes over the
alphabets ${\mathbb{F}}_8$ and ${\mathbb{F}}_{16}$.

\begin{example}
In this example we consider schemes over ${\mathbb{F}}_8$. Depending
on the error-probability it is sufficient to make $u_1+1$ or $q-1=7$
queries to correct an entry. The number of participants is
$n=8^2=64$. In Table~\ref{topstart1} -- Table~\ref{topslut1} we consider codes $C_1=RM_q(u_1,2)$ and
$C_2=RM_q(u_2)$ for different choices of $u_2 < u_1 \leq q-2$ and we
list the parameters $t=t_1, \ldots , t_\ell$ and $r_1, \ldots ,
r_\ell=r$ (in particular the number of columns equals the codimension
$\ell$). We also list corresponding numbers $t_1^\prime, \ldots ,
t_\ell^\prime$ and $r_1, \ldots , r_\ell^\prime$. They are the
lower bounds and upper bounds, respectively, that  we would get on the
$t_i$'s and the $r_i$'s, respectively, by using GHWs instead of RGHWs. It is quite clear that the amount of information leaked to
the party performing the local error-correction is often lower than
what could be anticipated from studying only the GHWs.

\begin{table}
\centering
\begin{tabular}{|c|c c c c c c c|}
\hline
$m$  & 1 & 2 & 3 & 4 & 5 & 6 & 7\\
\hline
$t_m$ & 6 & 12 & 17 & 21 & 24 & 26 & 27 \\
$t_m^\prime$ & 6 & 7 & 13 & 14 & 15 & 20 & 21 \\
$r_m$ & 22 & 24 & 27 & 31 & 36 & 42 & 49 \\
$r_m^\prime$ & 28 & 33 & 34 & 35 & 41 & 42 & 49 \\
\hline
\end{tabular}
\caption{Scheme based on $C_1 = RM_8(6,2)$ and $C_2 = RM_8(5,2)$. For local
error-correction 7 queries are needed.}
\label{topstart1}
\end{table}

\begin{table}
%\begin{adjustwidth}{-1in}{-1in}
\centering
\begin{tabular}{|c|c c c c c c c c c c c c c|}
\hline
$m$  & 1 & 2 & 3 & 4 & 5 & 6 & 7 & 8 & 9 & 10 & 11 & 12 & 13\\
\hline
$t_m$ & 5 & 6 & 11 & 12 & 16 & 17 & 20 & 21 & 23 & 24 & 25 & 26 & 27 \\
$t_m^\prime$ & 5 & 6 & 7 & 12 & 13 & 14 & 15 & 19 & 20 & 21 & 22 & 23 & 26 \\
$r_m$ & 16 & 17 & 19 & 20 & 23 & 24 & 28 & 29 & 34 & 35 & 41 & 42 & 49 \\
$r_m^\prime$ & 19 & 20 & 21 & 25 & 26 & 27 & 28 & 33 & 34 & 35 & 41 & 42 & 49 \\
\hline
\end{tabular}
\caption{Scheme based on $C_1 = RM_8(6,2)$ and $C_2 = RM_8(4,2)$. For local
error-correction 7 queries are needed.}
%\end{adjustwidth}
\end{table}
\begin{table}
\centering
\begin{tabular}{|c|c c c c c c |}
\hline
$m$  & 1 & 2 & 3 & 4 & 5 & 6 \\
\hline
$t_m$ & 5 & 10 & 14 & 17 & 19 & 20 \\
$t_m^\prime$ & 5 & 6 & 7 & 12 & 13 & 14 \\
$r_m$ & 16 & 19 & 23 & 28 & 34 & 41 \\
$r_m^\prime$ & 25 & 26 & 27 & 33 & 34 & 41 \\
\hline
\end{tabular}
\caption{Scheme based on $C_1 = RM_8(5,2)$ and $C_2 = RM_8(4,2)$. For local
error-correction 6 or 7 queries are needed, depending on the error-probability.}
\end{table}

\begin{table}
%\begin{adjustwidth}{-1in}{-1in}
\centering
\begin{tabular}{|c|c c c c c c c c c c c|}
\hline
$m$  & 1 & 2 & 3 & 4 & 5 & 6 & 7 & 8 & 9 & 10 & 11 \\
\hline
$t_m$ & 4 & 5 & 9 & 10 & 13 & 14 & 16 & 17 & 18 & 19 & 20 \\
$t_m^\prime$ & 4 & 5 & 6 & 7 & 11 & 12 & 13 & 14 & 15 & 18 & 19 \\
$r_m$ & 11 & 12 & 15 & 16 & 20 & 21 & 26 & 27 & 33 & 34 & 41 \\
$r_m^\prime$ & 13 & 17 & 18 & 19 & 20 & 25 & 26 & 27 & 33 & 34 & 41 \\
\hline
\end{tabular}
\caption{Scheme based on $C_1 = RM_8(5,2)$ and $C_2 = RM_8(3,2)$. For local
error-correction 6 or 7 queries are needed, depending on the error-probability.}
\end{table}
\begin{table}
\centering
\begin{tabular}{|c|c c c c c c c c c c c c c c c|}
\hline
$m$  & 1 & 2 & 3 & 4 & 5 & 6 & 7 & 8 & 9 & 10 & 11 & 12 & 13 & 14 & 15\\
\hline
$t_m$ & 3 & 4 & 5 & 8 & 9 & 10 & 12 & 13 & 14 & 15 & 16 & 17 & 18 & 19 & 20 \\
$t_m^\prime$ & 3 & 4 & 5 & 6 & 7 & 10 & 11 & 12 & 13 & 14 & 15 & 17 & 18 & 19 & 20 \\
$r_m$ & 7 & 8 & 9 & 12 & 13 & 14 & 18 & 19 & 20 & 25 & 26 & 27 & 33 & 34 & 41 \\
$r_m^\prime$ & 9 & 10 & 11 & 12 & 13 & 17 & 18 & 19 & 20 & 25 & 26 & 27 & 33 & 34 & 41 \\
\hline
\end{tabular}
\caption{Scheme based on $C_1 = RM_8(5,2)$ and $C_2 = RM_8(2,2)$. For local
error-correction 6 or 7 queries are needed, depending on the error-probability.}
\end{table}

\begin{table}
\centering
\begin{tabular}{|c|c c c c c |}
\hline
$m$  & 1 & 2 & 3 & 4 & 5\\
\hline
$t_m(RGHW)$ & 4 & 8 & 11 & 13 & 14 \\
$t_m(GHW)$ & 4 & 5 & 6 & 7 & 11 \\
$r_m(RGHW)$ & 11 & 15 & 20 & 26 & 33 \\
$r_m(GHW)$ & 18 & 19 & 25 & 26 & 33 \\
\hline
\end{tabular}
\caption{Scheme based on $C_1 = RM_8(4,2)$ and $C_2 = RM_8(3,2)$. For
local error-correction 5 or 7 queries are needed, depending on the error-probability.}
\label{topslut1}
\end{table}

%\begin{table}
%\begin{adjustwidth}{-1in}{-1in}
%\centering
%\begin{tabular}{|c|c c c c c c c c c |}
%\hline
%$m$  & 1 & 2 & 3 & 4 & 5 & 6 & 7 & 8 & 9\\
%\hline
%$t_m(RGHW)$ & 3 & 4 & 7 & 8 & 10 & 11 & 12 & 13 & 14 \\
%$t_m(GHW)$ & 3 & 4 & 5 & 6 & 7 & 10 & 11 & 12 & 13 \\
%$r_m(RGHW)$ & 7 & 8 & 12 & 13 & 18 & 19 & 25 & 26 & 33 \\
%$r_m(GHW)$ & 10 & 11 & 12 & 17 & 18 & 19 & 25 & 26 & 33 \\
%\hline
%\end{tabular}
%\caption{The values $t_1,\ldots,t_\ell$ and $r_1,\ldots,r_\ell$ for a
%scheme based on $C_1 = RM_8(4,2)$ and $C_2 = RM_8(2,2)$. For local
%error-correction 5 or 7 querries are needed, depending on the error-probability%.}
%\end{adjustwidth}
%\end{table}

%\begin{table}
%\centering
%\begin{tabular}{|c|c c c c |}
%\hline
%$m$  & 1 & 2 & 3 & 4\\
%\hline
%$t_m(RGHW)$ & 3 & 6 & 8 & 9 \\
%$t_m(GHW)$ & 3 & 4 & 5 & 6 \\
%$r_m(RGHW)$ & 7 & 12 & 18 & 25 \\
%$r_m(GHW)$ & 11 & 17 & 18 & 25 \\
%\hline
%\end{tabular}
%\caption{The values $t_1,\ldots,t_\ell$ and $r_1,\ldots,r_\ell$ for a
%scheme based on $C_1 = RM_8(3,2)$ and $C_2 = RM_8(2,2)$. For local
%error-correction 4 or 7 querries are needed, depending on the error-probability%.}\end{table}
\end{example}

\begin{example}
In this example we consider schemes over ${\mathbb{F}}_{16}$. Depending
on the error-probability it is sufficient to make $u_1+1$ or $q-1=15$
queries. The number of participants is $n=16^2=264$. The information
in Table~\ref{topstart2} -- Table \ref{topslut2} is similar to the previous example.

\begin{table}
%\begin{adjustwidth}{-1in}{-1in}
\centering
\begin{tabular}{|c|c c c c c c c c c c c c c c c|}
\hline
$m$  & 1 & 2 & 3 & 4 & 5 & 6 & 7 & 8 & 9 & 10 & 11 & 12 & 13 & 14 & 15\\
\hline
$t_m$ & 14 & 28 & 41 & 53 & 64 & 74 & 83 & 91 & 98 & 104 & 109 & 113 & 116 & 118 & 119 \\
$t_m^\prime$ & 14 & 15 & 29 & 30 & 31 & 44 & 45 & 46 & 47 & 59 & 60 & 61 & 62 & 63 & 74 \\
$r_m$ & 106 & 108 & 111 & 115 & 120 & 126 & 133 & 141 & 150 & 160 & 171 & 183 & 196 & 210 & 225 \\
$r_m^\prime$ & 161 & 162 & 163 & 164 & 165 & 177 & 178 & 179 & 180 & 193 & 194 & 195 & 209 & 210 & 225 \\
\hline
\end{tabular}
\caption{Scheme based on $C_1 = RM_{16}(14,2)$ and $C_2 = RM_{16}(13,2)$. For local
error-correction 15 queries are needed.}
%\end{adjustwidth}
\label{topstart2}
\end{table}

\begin{table}
\centering
\begin{tabular}{|c|c c c c c c c c c c c c c c|}
\hline
$m$  & 1 & 2 & 3 & 4 & 5 & 6 & 7 & 8 & 9 & 10 & 11 & 12 & 13 & 14\\
\hline
$t_m$ & 13 & 26 & 38 & 49 & 59 & 68 & 76 & 83 & 89 & 94 & 98 & 101 & 103 & 104 \\
$t_m^\prime$ & 13 & 14 & 15 & 28 & 29 & 30 & 31 & 43 & 44 & 45 & 46 & 47 & 58 & 59 \\
$r_m$ & 92 & 95 & 99 & 104 & 110 & 117 & 125 & 134 & 144 & 155 & 167 & 180 & 194 & 209 \\
$r_m^\prime$ & 146 & 147 & 148 & 149 & 161 & 162 & 163 & 164 & 177 & 178 & 179 & 193 & 194 & 209 \\
\hline
\end{tabular}
\caption{Scheme based on $C_1 = RM_{16}(13,2)$ and $C_2 = RM_{16}(12,2)$. For local
error-correction 14 or 15 queries are needed, depending on the error-probability.}
\end{table}

\begin{table}
\centering
\begin{tabular}{|c|c c c c c c c c c c c c c |}
\hline
$m$  & 1 & 2 & 3 & 4 & 5 & 6 & 7 & 8 & 9 & 10 & 11 & 12 & 13\\
\hline
$t_m$ & 12 & 24 & 35 & 45 & 54 & 62 & 69 & 75 & 80 & 84 & 87 & 89 & 90 \\
$t_m^\prime$ & 12 & 13 & 14 & 15 & 27 & 28 & 29 & 30 & 31 & 42 & 43 & 44 & 45 \\
$r_m$ & 79 & 83 & 88 & 94 & 101 & 109 & 118 & 128 & 139 & 151 & 164 & 178 & 193 \\
$r_m^\prime$ & 131 & 132 & 133 & 145 & 146 & 147 & 148 & 161 & 162 & 163 & 177 & 178 & 193 \\
\hline
\end{tabular}
\caption{Scheme based on $C_1 = RM_{16}(12,2)$ and $C_2 = RM_{16}(11,2)$. For local
error-correction 13 or 15 queries are needed, depending on the error-probability.}
\end{table}

\begin{table}
\centering
\begin{tabular}{|c|c c c c c c c c c c c c|}
\hline
$m$  & 1 & 2 & 3 & 4 & 5 & 6 & 7 & 8 & 9 & 10 & 11 & 12\\
\hline
$t_m$ & 11 & 22 & 32 & 41 & 49 & 56 & 62 & 67 & 71 & 74 & 76 & 77 \\
$t_m^\prime$ & 11 & 12 & 13 & 14 & 15 & 26 & 27 & 28 & 29 & 30 & 31 & 41 \\
$r_m$ & 67 & 72 & 78 & 85 & 93 & 102 & 112 & 123 & 135 & 148 & 162 & 177 \\
$r_m^\prime$ & 116 & 117 & 129 & 130 & 131 & 132 & 145 & 146 & 147 & 161 & 162 & 177 \\
\hline
\end{tabular}
\caption{Scheme based on $C_1 = RM_{16}(11,2)$ and $C_2 = RM_{16}(10,2)$. For local
error-correction 12 or 15 queries are needed, depending on the error-probability.}
\end{table}

\begin{table}
\centering
\begin{tabular}{|c|c c c c c c c c c c c|}
\hline
$m$  & 1 & 2 & 3 & 4 & 5 & 6 & 7 & 8 & 9 & 10 & 11\\
\hline
$t_m$ & 10 & 20 & 29 & 37 & 44 & 50 & 55 & 59 & 62 & 64 & 65 \\
$t_m^\prime$ & 10 & 11 & 12 & 13 & 14 & 15 & 25 & 26 & 27 & 28 & 29 \\
$r_m$ & 56 & 62 & 69 & 77 & 86 & 96 & 107 & 119 & 132 & 146 & 161 \\
$r_m^\prime$ & 101 & 113 & 114 & 115 & 116 & 129 & 130 & 131 & 145 & 146 & 161 \\
\hline
\end{tabular}
\caption{Scheme based on $C_1 = RM_{16}(10,2)$ and $C_2 = RM_{16}(9,2)$. For local
error-correction 11 or 15 queries are needed, depending on the error-probability.}
\end{table}

\begin{table}
\centering
\begin{tabular}{|c|c c c c c c c c c c|}
\hline
$m$  & 1 & 2 & 3 & 4 & 5 & 6 & 7 & 8 & 9 & 10\\
\hline
$t_m$ & 9 & 18 & 26 & 33 & 39 & 44 & 48 & 51 & 53 & 54 \\
$t_m^\prime$ & 9 & 10 & 11 & 12 & 13 & 14 & 15 & 24 & 25 & 26 \\
$r_m$ & 46 & 53 & 61 & 70 & 80 & 91 & 103 & 116 & 130 & 145 \\
$r_m^\prime$ & 97 & 98 & 99 & 100 & 113 & 114 & 115 & 129 & 130 & 145 \\
\hline
\end{tabular}
\caption{Scheme based on $C_1 = RM_{16}(9,2)$ and $C_2 = RM_{16}(8,2)$. For local
error-correction 10 or 15 queries are needed, depending on the
error-probability.}
\label{topslut2}
\end{table}

\end{example}

\section{Concluding remarks}\label{secconcl}\label{secosaco8}
In this paper we applied a coset
construction of $q$-ary Reed-Muller codes to the situation where a
central party wants to store a secret on a distributed media in such a
way that other parties with access to a large part of the media can
recover the secret, whereas parties with limited access cannot. The
reason for choosing $q$-ary Reed-Muller codes is that with such codes
one is able to perform local error-correction. For the purpose of
analysing the information leakage we determined the relative
generalized Hamming weights of the codes involved. This was done using
the footprint bound from Gr\"{o}bner basis theory. There is a very strong connection between the
footprint bound and the Feng-Rao bound for primary codes~\cite{andersen2008evaluation,geil2013feng}
which is the bound that we used in~\cite{geil2014} to estimate RGHWs of
one-point algebraic geometric codes. Using the
footprint bound rather than the Feng-Rao bound for primary or dual codes
saved us some cumbersome notation (which is difficult to avoid in
the case of one-point algebraic geometric codes). Using
the derived information on the RGHWs we discussed the trade off
between security in the above scheme and the ability to perform
local error-correction.\\

\section*{Acknowledgments}

The authors gratefully acknowledge the support from
the Danish National Research Foundation and the National Natural Science
Foundation of China (Grant No.\ 11061130539) for the Danish-Chinese
Center for Applications of Algebraic Geometry in Coding Theory and
Cryptography. Also the authors gratefully acknowledge the support from
The Danish Council for Independent Research (Grant
No.\ DFF--4002-00367). 
Part of this work was done while the first listed author
was visiting East China Normal University. We are grateful to
Professor Hao Chen for his hospitality. Finally the authors would like
to thank Diego Ruano, Hans H\"{u}ttel and Ruud Pellikaan for helpful discussions.

\appendix

\section{Proof of Lemma~\ref{lem5}} \label{A1}
To prove Lemma~\ref{lem5} we start by
generalizing~\cite[Th.\ 3.7.7]{heijnen1999phd} which corresponds to 
Lemma~\ref{lem3} below in the particular case that $b=s(q-1)$. The
proof of \cite[Th.\ 3.7.7]{heijnen1999phd} was given in
\cite[App.\ B.1]{heijnen1999phd}. 

\begin{lemma} \label{lem3}
Let $A$ be a subset of $F_q(a,b)$ consisting of $m$ elements. Then $|\Delta L_{(a,b)}(m)| \leq |\Delta A|$.
\end{lemma}

\begin{proof}
In Appendix B.1 of \cite{heijnen1999phd} a proof for Lemma~\ref{lem3}
is given in the particular case that $b=s(q-1)$. We indicate how this
proof can be modified to cover all possible choices of $b$. First note
that~\cite{heijnen1999phd} uses $v$ where we use $a$, uses $m$ where
we use $s$, and uses $r$ where we use $m$. With the following
modifications the proof in~\cite{heijnen1999phd} is lifted to a proof of Lemma~\ref{lem3}.
\begin{itemize}
\item In \cite[Rem.\ B.1.2]{heijnen1999phd}: Replace $F_{\geq v}$ with
  $F_q(v,b)$ and let the parameter $k$ go from $v$ to $b$. 
\item In \cite[Def.\ B.1.6]{heijnen1999phd}: Replace $F_{\geq l}$ with
  $F_q(l,b)$. 
\item In \cite[Lem.\ B.1.10]{heijnen1999phd}: Replace $F_{\geq v}$
  with $F_q(v,b)$ and let the summation end with $A_b$ rather than
  $A_{s(q-1)}$. 
\item In \cite[Lem.\ B.1.13, Lem.\ B.1.14 and their proofs]{heijnen1999phd}: Replace  $F_{\geq l}$, $F_{\geq (l-1)}$,
  $F_{\geq v}$, $L_{\geq l-1}(r)$ and $L_{\geq l}(r)$ with $F_q(l,b)$,
  $F_q(l-1,b)$, $F_q(v,b)$, $L_{(l,b)}(r)$ and $L_{(l-1,b)}(r)$, respectively.
\end{itemize}
\end{proof}

Recall from Section~\ref{secthree} the map $\mu: Q^s_q \rightarrow Q^s_q$ given by
$\mu(a_1,\ldots,a_s) = (q-1-a_s,\ldots,q-1-a_1)$. To translate
Lemma~\ref{lem3} into Lemma~\ref{lem5} we need the following results.

\begin{lemma} \label{lem4}
Let $0 \leq a \leq b \leq s(q-1)$ be integers,  $\vec{a}, \vec{b} \in Q^s_q$ and $m \in \{1,\ldots,|F_q(a,b)|\}$, then we have that
\begin{enumerate}
\item $\vec{a} \lexprec \vec{b} \iff \mu(\vec{a}) \antiprec \mu(\vec{b})$,
\item $\vec{a} \antiprec \vec{b} \iff \mu(\vec{a}) \lexprec \mu(\vec{b})$,
\item $\vec{a} \Ppreceq \vec{b} \iff \mu(\vec{a}) \Psucceq \mu(\vec{b})$,
\item $\mu(\nabla \vec{a}) = \Delta \mu(\vec{a})$,
\item $\mu(\nabla A) = \Delta \mu(A)$, 
\item $\mu(F_q(a,b)) = F_q(s(q-1)-b,s(q-1)-a)$,
\item $A \subseteq F_q(a,b) \iff \mu(A) \subseteq F_q(s(q-1)-b,s(q-1)-a)$,
\item $\mu( N_{(a,b)}(m)) = L_{(s(q-1)-b,s(q-1)-a)}(m)$,
\item $\mu(\nabla N_{(a,b)}(m)) = \Delta  L_{(s(q-1)-b,s(q-1)-a)}(m)$.
\end{enumerate}
\end{lemma}

\begin{proof}

Let $\vec{a} = (a_1,\ldots,a_s)$ and $\vec{b} = (b_1,\ldots,b_s)$.
\begin{enumerate}
\item $ \vec{a} \lexprec \vec{b} \iff  a_1 = b_1, \ldots, a_{l-1} = b_{l-1}, a_l < b_l\mbox{ for some }l \iff q-1-a_1 = q-1-b_1, \ldots, q-1-a_{l-1} = q-1-b_{l-1}, q-1-a_l > q-1-b_l\mbox{ for some }l \iff \mu(\vec{a}) \antiprec \mu(\vec{b}) $.

\item Similar to 1.

\item $ \vec{a} \Ppreceq \vec{b} \iff   a_1 \leq b_1, \ldots, a_s \leq b_s  \iff q-1-a_1 \geq q-1-b_1, \ldots, q-1-a_{s} \geq q-1-b_{s} \iff   \mu(\vec{a}) \Psucceq \mu(\vec{b})$.

\item $ \vec{b} \in \mu(\nabla \vec{a})  \iff   \exists \vec{b}_1 = \mu^{-1}(b) \in \nabla \vec{a} \iff \vec{b}_1 \Ppreceq \vec{a} \iff   \mu(\vec{b}_1) \Psucceq \mu(\vec{a}) \iff \vec{b} \Psucceq \mu(\vec{a}) \iff   \vec{b} \in \Delta \mu(\vec{a})$.

\item $\mu(\nabla A) = \mu(\bigcup_{\vec{a} \in A}\nabla \vec{a}) = \bigcup_{\vec{a} \in A} \mu(\nabla \vec{a}) = \bigcup_{\vec{a} \in A} \Delta \mu(\vec{a}) = \Delta \bigcup_{\vec{a} \in A} \mu(\vec{a})  =   \Delta \mu(A)$.

\item $\vec{a} \in F_q(a,b) \iff a \leq \deg(\vec{a}) \leq b \iff   a \leq \sum_{i=1}^s a_i  \leq b \iff s(q-1)-a \geq s(q-1)-\sum_{i=1}^s a_i \geq s(q-1)-b  \iff   s(q-1)-b \leq \sum_{i=1}^s (q-1-a_i) \leq s(q-1)-a \iff  \mu(\vec{a}) \in F_q(s(q-1)-b,s(q-1)-a ) $.

\item Similar to 6.

\item Follows from 1,2 and 7 by induction.

\item $\mu(\nabla N_{(a,b)}(m)) = \Delta \mu(N_{(a,b)}(m)) = \Delta L_{(s(q-1)-b,s(q-1)-a)}(m)$.
\end{enumerate}
\end{proof}

We are now ready to prove Lemma~\ref{lem5}.

\begin{proof}[Proof of Lemma \ref{lem5}]
By 7.\ in Lemma \ref{lem4} we have $\mu(A) \subseteq F_q(s(q-1)-b,s(q-1)-a)$. It follows that
\begin{eqnarray*}
|\nabla N_{(a,b)}(m)| & = & |\mu(\nabla N_{(a,b)}(m))| \\
& = & |\Delta L_{(s(q-1)-b,s(q-1)-a)}(m)| \\
& \leq & |\Delta \mu(A)| \\
& = & |\mu(\nabla A)| \\
& = & |\nabla A|, \\
\end{eqnarray*}
where the first and the last line is a consequence of the fact that $\mu$ is
bijective, the second line follows from 9.\ in Lemma~\ref{lem4}, the third line follows from Lemma~\ref{lem3}, and the fourth line
follows from 5.\ in Lemma~\ref{lem4}.
\end{proof}

\bibliographystyle{abbrv}

\end{document}